\documentclass[11pt,a4paper]{article}
\pdfoutput=1

\usepackage[margin=2.5cm]{geometry}
\usepackage{amsmath,amssymb,thmtools,amsthm, booktabs,color,doi,graphicx,latexsym,url,xcolor,xspace}
\usepackage{thm-restate}
\usepackage[numbers,sort&compress]{natbib}
\usepackage{microtype,hyperref}
\usepackage{eucal}
\usepackage[inline]{enumitem}
\setlist[itemize]{label=--} 
\setlist[enumerate]{label=(\arabic*),labelindent=\parindent,leftmargin=*}

\definecolor{citecolor}{HTML}{0000C0}
\definecolor{urlcolor}{HTML}{000080}

\usepackage[utf8]{inputenc}

\newtheorem{theorem}{Theorem}
\newtheorem{corollary}[theorem]{Corollary}

\newtheorem{lemma}[theorem]{Lemma}
\newtheorem{proposition}[theorem]{Proposition}

\DeclareMathOperator*{\radius}{rad}

\DeclareMathOperator*{\diam}{diam}

\newcommand{\none}{\textsc{none}}

\newcommand{\St}{\mathbb{G}_t}
\newcommand{\Sr}{\mathbb{G}_r}
\newcommand{\Sto}{\mathbb{G}_{t+1}}
\newcommand{\Sro}{\mathbb{G}_{r+1}}
\newcommand{\Xt}[1]{\mathbb{X}_t(#1)}

\newcommand{\Rt}[1]{\mathbb{R}_t(#1)}
\newcommand{\Nt}[1]{\mathbb{N}_t(#1)}
\newcommand{\Rto}[1]{\mathbb{R}_{t+1}(#1)}
\newcommand{\Tt}[1]{\mathbb{T}_t(#1)}
\newcommand{\Tto}[1]{\mathbb{T}_{t+1}(#1)}
\newcommand{\Tr}[1]{\mathbb{T}_r(#1)}

\newcommand{\config}[1]{{\mathcal{#1}}}

\newcommand{\namedref}[2]{\hyperref[#2]{#1~\ref*{#2}}}
\newcommand{\sectionref}[1]{\namedref{Section}{#1}}
\newcommand{\appendixref}[1]{\namedref{Appendix}{#1}}
\newcommand{\figureref}[1]{\namedref{Figure}{#1}}

\newcommand{\equationref}[1]{\hyperref[#1]{Eq~(\ref*{#1})}}
\newcommand{\theoremref}[1]{\hyperref[#1]{Theorem~\ref*{#1}}}
\newcommand{\lemmaref}[1]{\hyperref[#1]{Lemma~\ref*{#1}}}
\newcommand{\remarkref}[1]{\hyperref[#1]{Remark~\ref*{#1}}}
\newcommand{\definitionref}[1]{\hyperref[#1]{Definition~\ref*{#1}}}
\newcommand{\corollaryref}[1]{\hyperref[#1]{Corollary~\ref*{#1}}}

\hypersetup{
    colorlinks=true,
    linkcolor=black,
    citecolor=citecolor,
    filecolor=black,
    urlcolor=urlcolor,
}

\newenvironment{mycover}
{\list{}{\listparindent 0pt
        \itemindent    \listparindent
        \leftmargin    0cm
        \rightmargin   0cm
        \parsep        0pt}%
    \raggedright
    \item\relax}
{\endlist}

\newcommand{\myemail}[1]{\,$\cdot$\, {\small #1}}
\newcommand{\myaff}[1]{\,$\cdot$\, {\small #1}\par\medskip}

\title{Wait-free approximate agreement on graphs}

\begin{document}

\begin{mycover}
    {\huge\bfseries\boldmath Wait-free approximate agreement on graphs \par}
    \bigskip
    \bigskip
    \textbf{Dan Alistarh}
    \myemail{dan.alistarh@ist.ac.at}
    \myaff{IST Austria}

    \textbf{Faith Ellen}
    \myemail{faith@cs.toronto.edu}
    \myaff{University of Toronto}

    \textbf{Joel Rybicki}
    \myemail{joel.rybicki@ist.ac.at}
    \myaff{IST Austria}

\end{mycover}

\medskip
\noindent\textbf{Abstract.}
  Approximate agreement is one of the few variants of consensus that can be solved in a wait-free manner in asynchronous systems where processes communicate by reading and writing to shared memory.
  In this work, we consider a natural generalisation of approximate agreement on arbitrary undirected connected graphs.
Each process is given a vertex of the graph as input and, if non-faulty, must output a vertex such that
\begin{itemize}[noitemsep]
\item all the outputs are within distance 1 of one another, and 
\item each output value lies on a 
shortest path 
between two input values.
\end{itemize}
From prior work, it is known that there is no wait-free algorithm among $n \ge 3$ processes for this problem on any cycle of length $c \ge 4$, by reduction from 2-set agreement (Casta\~neda et al., 2018).

In this work, we investigate the solvability and complexity of this task on general graphs.
We give a new, direct proof of the impossibility of approximate agreement on cycles of length $c \ge 4$, via a generalisation of Sperner's Lemma to convex polygons. We also extend the reduction from 2-set agreement to a larger class of graphs, showing that approximate agreement on on these graphs is unsolvable. 
Furthermore, we show that combinatorial arguments, used by both existing proofs, are necessary, by showing that the impossibility of a wait-free algorithm in the nonuniform iterated snapshot model cannot be proved via an extension-based proof. On the positive side, we present a wait-free algorithm for a class of graphs that properly contains the class of chordal graphs.

\thispagestyle{empty}
\setcounter{page}{0}
\newpage

\section{Introduction}

Understanding the solvability and complexity of coordination tasks is one of the key questions in distributed computing.
The difficulty of coordination often arises from \emph{uncertainty}: processes have limited knowledge about each other's inputs, the relative speed of computation and communication between processes can vary, and processes may fail during computation.

Tasks which require perfect agreement, such as 
\emph{consensus}~\cite{pease80reaching},
are typically hard to solve:
Fischer, Lynch, and Paterson~\cite{fischer85impossibility} proved that consensus cannot be reached in asynchronous message-passing systems if even one process may crash. Later, this was extended to shared memory systems where processes communicate 
using shared registers~\cite{chor1987processor,loui1987memory}.

While perfect agreement is 
not needed for many applications, it is known that agreeing
on at most $k > 1$ different values is still hard:
There exists no algorithm for $k$-set agreement that tolerates $k$ crash faults in the asynchronous setting for $n > k$ processes~\cite{borowsky1993generalized,herlihy1999topological,saks2000wait}. In contrast, approximate agreement -- agreeing on values that are sufficiently close to one another -- can be considerably easier~\cite{Dolev1986Reaching,attiya1994wait-free,Schenk1995,fekete1990asymptotically,fekete1994asynchronous,mendes2014distributed}. 

\subsection{Graphical approximate agreement}
In this work, we study solvability and complexity of approximate agreement when the set of input and output values reside on a graph. Consider a distributed system with $n$ processes 
and let $G = (V,E)$ be a connected graph.
The graph $G$ is not assumed to be related to the communication topology of the distributed system, but it is assumed to be known by all processes. In approximate agreement on~$G$, each process $p_i$ is given a node $x_i \in V$ as input and has to  output a node $y_i \in V$ subject to the following constraints: 
\begin{itemize}[noitemsep]
\item agreement: every two output values are adjacent in $G$,  and
\item (shortest path) validity: each output value lies on a shortest path between two input values. 
\end{itemize}
Note that the output values form a clique. \figureref{fig:examples}(a) gives an example of graphical approximate on a tree. Prior work has mostly focused on 
the cases
when $G$ is 
a path~\cite{Dolev1986Reaching,attiya1994wait-free,Schenk1995,fekete1990asymptotically,fekete1994asynchronous}, 
a graph whose clique graph is a tree~\cite{alcantara2019topology}, or a chordal graph~\cite{nowak2019byzantine}, i.e., a graph that contains no induced cycle of length greater than three.

\begin{figure}[t]
\begin{center}
  \includegraphics[page=1,width=0.9\textwidth]{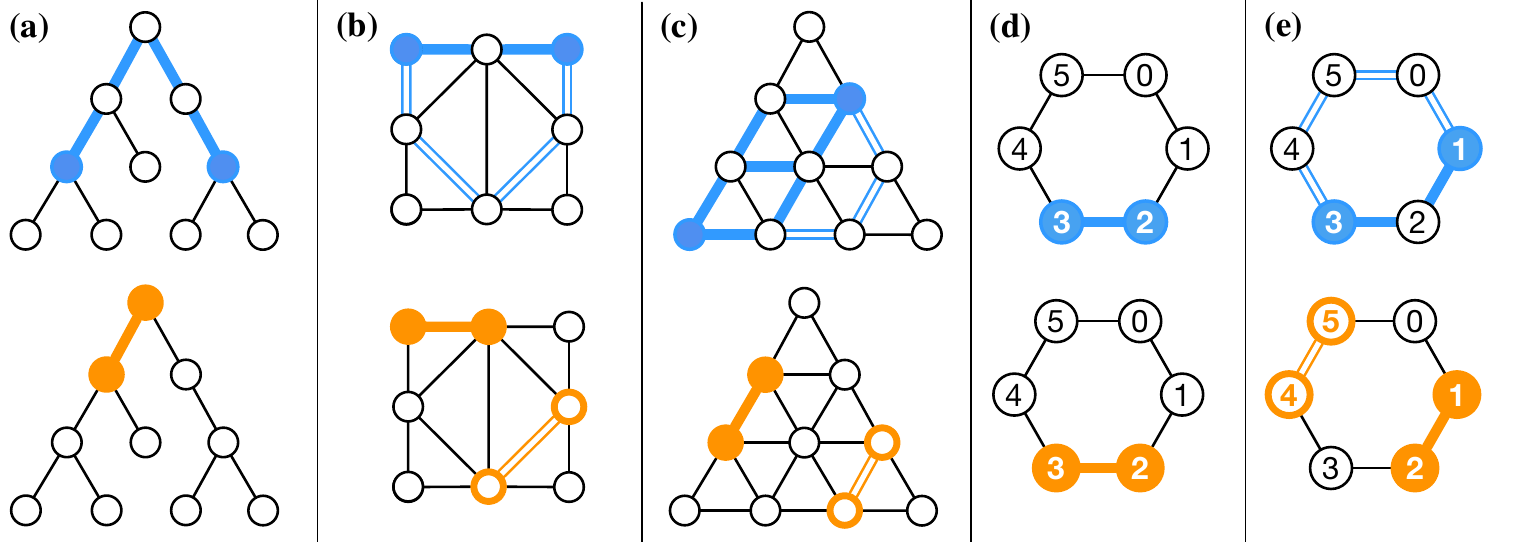}
  \caption{Examples of approximate agreement 
  with $n=2$ processes.
  In the top row, blue nodes are input values for a particular instance.
 Solid blue edges denote edges on \emph{shortest} paths, and blue double lines denote the additional edges that are also on some \emph{minimal} path connecting 
 input nodes. Solid and non-solid orange nodes denote outputs that satisfy the shortest path and minimal path validity constraints, respectively.
  (a)~Agreement on a tree. All minimal paths are also shortest paths. (b)~Agreement on a chordal graph. (c)~Agreement on a non-chordal bridged graph. (d)--(e) Instances of 6-cycle agreement.
  \label{fig:examples}}
\end{center}
\end{figure}

\paragraph{Approximate agreement on a path.} The special case when $G$ is a path is well-understood. This case is typically studied in the continuous setting, where the values reside on the real line and the goal is to output values within distance $\varepsilon > 0$ of each other.
A discrete version of the problem can be obtained by considering integer-valued inputs and outputs and taking $\varepsilon = 1$.
In the shared-memory setting, Attiya, Lynch, and Shavit~\cite{attiya1994wait-free} showed that the step complexity of wait-free solutions using single-writer registers is $\Theta(\log n)$. Using multi-writer registers, Schenk~\cite{Schenk1995} established that the step complexity of obtaining agreement 
is $O(\log D)$, where $D$ is the maximum distance between two input values.

In asynchronous message-passing systems, Dolev, Lynch, Pinter, Stark and Weihl~\cite{Dolev1986Reaching} 
showed that approximate agreement can be solved with $f < n/5$ Byzantine faults.
This was improved by Abraham, Amit, and Dolev~\cite{Abraham2005optimal} to allow $f <n/3$ Byzantine faults, matching a lower bound by Fischer, Lynch, and Merritt~\cite{fischer1986easy}.
Efficient algorithms tolerating more benign faults in the synchronous and asynchronous message-passing settings were given by Fekete~\cite{fekete1990asymptotically,fekete1994asynchronous}.

\paragraph{Approximate agreement under minimal path validity.}
Rybicki and Nowak~\cite{nowak2019byzantine} studied approximate agreement on chordal graphs under a slightly different validity condition, where output values have to lie on a \emph{minimal} path between any two input values. A path in $G$ is \emph{minimal} if no two non-consecutive nodes in the path are connected by an edge, i.e.~if $v_0,\ldots,v_k$ is a minimal path and $0 \leq i < j-1 \leq k-1$, then $\{v_i,v_j\} \not\in E$. This validity condition is weaker, since every shortest path between two nodes is a minimal path, but the converse is not true. Figures~\ref{fig:examples}(b)--(c) illustrate the difference between minimal and shortest paths.
If $G$ is chordal, then there exists an  
algorithm tolerating $f$ \emph{Byzantine} faults in the asynchronous message-passing model for $n > (\omega(G)+1)f$ processes, where $\omega(G)$ is the size of the largest clique in $G$~\cite{nowak2019byzantine}.

\paragraph{Approximate gathering on graphs.}
Alc{\'a}ntara, Casta{\~n}eda, Flores-Pe{\~n}aloza, and Rajsbaum~\cite{alcantara2019topology} investigated 
approximate agreement
with the following weaker \emph{clique gathering} validity condition: if all inputs values are adjacent, then each output value has to be one of the input values. Their validity condition arises from considering an approximate gathering problem for robots on a graph.
This condition is weaker than shortest path and minimal path validity: for example, in the instances given in Figures~\ref{fig:examples}(b)--(e), any set of outputs that lie on a clique would satisfy clique gathering validity.

They showed that this problem is solvable in a wait-free manner on graphs whose clique graphs are trees and on graphs of radius one (i.e., graphs with a dominating set of size one). 
A clique graph $K(G)$ of $G$ is the graph where vertices of $K(G)$ are the maximal cliques of $G$ and two vertices of $K(G)$ are adjacent if they correspond to cliques with a common vertex. 
Note that there are chordal graphs whose clique graphs are not trees; for example, see \figureref{fig:examples}(b). 

\paragraph{Approximate agreement on cycles.}
When $G$ is a cycle of length $c \ge 4$,
approximate agreement under
minimal path validity 
and clique gathering validity
are the same problem.
We refer to this special case as \emph{$c$-cycle agreement}. 
When $c=3$, the problem is trivial, since each process can output its input.

Casta{\~n}eda, Rajsbaum, and Roy~\cite{castaneda2018convergence} showed that 2-set agreement 
reduces to $c$-cycle agreement,
for $c \ge 4$.
This implies that there is no algorithm for 
approximate agreement on $c$-cycles (under both minimal and shortest path validity) for $c \geq 4$
that tolerates $2$ crash faults in 
in asynchronous shared memory systems consisting of registers.
Hence, approximate agreement on cycles of length at least 4 is harder than on paths and chordal graphs.

\subsection{Contributions}

In this work, we establish
additional
positive and negative results on the solvability and complexity of graphical approximate agreement.

\paragraph{Positive results.}
We present a wait-free asynchronous algorithm 
for $n \geq 2$ processes that solves approximate agreement
on a large subclass of
bridged graphs, and on any radius one graph.
A \emph{bridged graph} is a graph in which each of its
cycles of length at least 4 contains 2 vertices that are connected by a shorter path than either path in the cycle connecting them~\cite{Farber1987Local,farber1989diameters}.
All chordal graphs are bridged, but the converse is not necessarily true; for an example, see \figureref{fig:examples}(c).

Our algorithm solves the graphical approximate agreement problem on all chordal graphs and a large class of non-chordal graphs of arbitrary large radius.
This includes graphs of radius one and graphs whose clique graphs are trees.
Thus, our algorithm handles all graphs handled by previous algorithms, while guaranteeing a stronger validity condition.
See Table 1 for a~comparison.

In addition, we give a 1-resilient asynchronous algorithm for graphical approximate agreement 
using only registers
on \emph{any} connected graph for $n \ge 2$ processes. 
Note that, when $n=2$, this algorithm is wait-free.
For the fully-connected synchronous message-passing model, 
we also
present
an $f$-resilient synchronous algorithm for the fully-connected message-passing model with $n>f$ processes. The algorithm solves approximate agreement on any connected graph $G$ in $\lfloor f/2 \rfloor + \lceil \log_2 \diam(G) \rceil + 1$ rounds, where $\diam(G)$ is the diameter of $G$.

\begin{table}[t]
\center
\begin{tabular}{@{}llllll@{}}
\toprule
Graph class & Validity condition & Fault model & Reference   \\ 
\midrule
Clique graph is a tree  &  Clique gathering  &  Wait-free & \cite{alcantara2019topology}      \\
Radius one &   &  Wait-free & \cite{alcantara2019topology}  \\
\midrule
Chordal  &  Minimal paths & Byzantine, $n > (\omega+1)f$  & \cite{nowak2019byzantine}  \\
\midrule
Paths & Shortest paths & Wait-free & \cite{attiya1994wait-free,Schenk1995} \\
Paths &  & Byzantine, $n > 3f$ & \cite{Abraham2005optimal} \\
Nicely bridged or radius one &  & Wait-free  &  {\bf this work}  \\ 
Any & & 1-resilient & {\bf this work} \\
\bottomrule
\end{tabular}
\caption{Algorithms for asynchronous approximate agreement on graphs.}
\label{table:summary}
\end{table}

\paragraph{Negative results.}
We provide a new, direct proof of the impossibility of approximate agreement on cycles of length
$c \geq 4$. It uses a generalisation of Sperner's Lemma to convex polygons.
It follows from known simulation techniques~\cite{gafni1998round,borowsky2001bg} that
there is no 2-resilient asynchronous algorithm using registers  and
any $f$-resilient synchronous algorithm requires at least $\lfloor f/2 \rfloor +1 $ rounds for $n > f$ processes. Furthermore, we present a simplified version of the existing reduction
from $k$-set agreement to cycle agreement and use it to
extend the impossibility of graphical approximate agreement 
to a larger class of
graphs.

\paragraph{Extension-based proofs.} Finally, in \sectionref{sec:extension}, we show that extension-based proofs~\cite{alistarh2019extension},
such as valency arguments, are not powerful enough to  show the impossibility of 4-cycle agreement in the non-uniform iterated snapshot model.
Note that this result does not follow from the fact that there are no extension-based proofs of the impossibility of 2-set agreement in the 
non-uniform iterated snapshot model~\cite{AAEGZ20}, even though there is a reduction from 2-set agreement to $c$-cycle agreement for $c \geq 4$.

\section{Related work}

\paragraph{Multidimensional approximate agreement.}
Mendes, Herlihy, Vaidya and Garg~\cite{MHVG15} generalised approximate agreement to the multidimensional setting, where the input values are points in $m$-dimensional Euclidean space $\mathbb{R}^m$, for $m > 0$. In the multidimensional approximate agreement problem, the 
output values should be within distance $\varepsilon > 0$ of one another
and be contained in the convex hull of 
the input values of the non-faulty processes.
When $m=1$, this is
approximate agreement on a line.
Multidimensional approximate agreement on $\mathbb{R}^m$ is solvable with $f$ Byzantine faults in the asynchronous completely-connected message-passing setting if and only if $n > (m+2)f$~\cite{MHVG15}. In the synchronous setting, the problem is solvable if and only if $n > \max \{ 3f, (m+1)f \}$. Recently, F\"ugger and Nowak~\cite{fuegger2018fast} established asymptotically tight convergence rates for multidimensional approximate agreement by removing the dependence on the dimension $m$ of the space.

Unlike approximate agreement on the real line, 
it is not straightforward to obtain a discrete version of multidimensional approximate agreement when $m \geq 2$.
For example, in the two-dimensional integer lattice $\mathbb{Z}^2 \subseteq \mathbb{R}^2$, one can find a pair of points arbitrarily far apart such that they are the only integral points in their convex hull.
In this case, solving approximate agreement is the same as solving consensus. More generally, Herlihy and Shavit~\cite{herlihy1993asynchronous}
showed that approximate agreement in a multidimensional setting with Euclidean convex hulls cannot be solved in a wait-free manner
when processes communicate using registers
if the space of values has holes of size $\varepsilon$. Since the Euclidean convex hull of two antipodal points around the hole consists of only the two points, outputting values within distance $\varepsilon$ of one another in this convex hull would amount to solving consensus.

Barycentric agreement~\cite{herlihy2013book} is a multidimensional problem that can be solved wait-free manner: processes are given inputs that lie on a simplex $\sigma$ of a simplicial complex and must output values that are on a simplex of the barycentric subdivision of $\sigma$. This problem can be solved, for example, using $m$-dimensional approximate agreement~\cite{mendes2014distributed}.

\paragraph{Approximate robot gathering in graphs.}
Robot gathering problems have been studied in the continuous setting~\cite{agmon2006fault,cieliebak2012distributed},
but we focus on the discrete setting, where $n$ robots reside on nodes in a graph $G$.
The inputs represent the initial positions of the robots, 
the outputs represent the final positions of the robots, and
the goal is that the outputs are close to one another.

\emph{Exact} gathering of asynchronous robots, where the goal is to get all robots to the same vertex, has been studied extensively in various models. See a recent survey of Cicerone, Di Stefano, and Navarra~\cite{cicerone2019asynchronous}.
Casta{\~n}eda, Rajsbaum, and Roy~\cite{castaneda2018convergence} and Alc{\'a}ntara, Casta{\~n}eda, Flores-Pe{\~n}aloza, and Rajsbaum~\cite{alcantara2019topology} studied several variants of 
approximate
gathering of asynchronous robots moving on a graph that communicate via snapshots.
In \emph{edge gathering}~\cite[Definition 4]{alcantara2019topology}, agreement is satisfied if all outputs 
belong to
the same edge. Validity requires that (i) if all inputs values are the same, then the output values are the same as the input values, and (ii) if all inputs belong to the same edge, then the output values also belong to this edge. The \emph{1-gathering} task~\cite[Definition 5]{alcantara2019topology} is a relaxation of edge gathering, where agreement is satisfied if the output values form a clique, and validity requires that the output values must be a subset of the input values if the input values form a clique.

Note that neither edge gathering or 1-gathering  solve graphical approximate agreement, as the validity constraint of graphical approximate agreement is stronger: each output value has to lie on some shortest path between two input values. The difference is best illustrated by the simple case of a path, where approximate agreement requires that the outputs always
lie
between 
the minimal and maximal input values, while edge gathering and 1-gathering do not have this requirement. 

Edge gathering is solvable if and only if $G$ is a tree~\cite{alcantara2019topology}.
On cliques, edge gathering
is the same as the 2-set agreement task, whereas 1-gathering and graphical approximate agreement are trivial. For 1-gathering, Alc{\'a}ntara et al.~\cite{alcantara2019topology} gave an algorithm for trees, which can also be used to solve 1-gathering on any graph whose clique graph is a tree.

When the graph $G$ is a cycle of length $c \ge 4$, edge gathering and 1-gathering are the same as $c$-cycle agreement. Casta{\~n}eda et al.~\cite{castaneda2018convergence} and Alc{\'a}ntara et al.~\cite{alcantara2019topology} gave a clever reduction showing that this problem is as hard as 2-set agreement 
for $n=3$ processes. In \sectionref{sec:impossibility-wait-free},
we give a direct proof of this result. Moreover, 
in \sectionref{sec:reductions},
we simplify and adapt the reduction from 2-set agreement to 
prove that wait-free graphical approximate agreement is impossible on a much larger class of graphs.

\section{Models}

We consider distributed systems with $n$ processes, where some processes may fail by crashing.
In particular, we 
focus on the setting where processes communicate using 
atomic snapshot objects,
which can be implemented from registers.
We also consider the 
synchronous message-passing model under fully-connected communication topologies.

\subsection{Asynchronous shared memory models}

In the  \emph{f-resilient non-uniform iterated snapshot} ($f$-NIS) model, $n$ processes, $p_0,\dots,p_{n-1}$, communicate using an infinite sequence, $S_1,S_2,\dots$, of shared \emph{single-writer atomic snapshot} objects.
A \emph{single-writer atomic snapshot} object has $n$ components, each of which has initial value $-$. It supports two atomic operations, $\mathsf{update}$ and $\mathsf{scan}$.
An $\mathsf{update}(x)$ by process $p_i$ changes the value of component $i$ to $x \neq -$.
A $\mathsf{scan}$ returns the value of each component.

Each process performs an $\mathsf{update}$ on a snapshot object, starting with $S_1$, and then repeatedly performs $\mathsf{scan}$s of this object until at most $f$ components have value $-$. 
(Note that, if $f = n-1$, then one $\mathsf{scan}$ of the snapshot object suffices, since the process has already performed an $\mathsf{update}$ on its own component.)
Next, it updates its state and applies a function, $\Delta$, to its new state to determine whether it should output a value. If the value of $\Delta$ is not $\bot$, then $p_i$ outputs this value and terminates. If the value of $\Delta$ is $\bot$, then, at its next step, it $\mathsf{update}$s the next snapshot object in the sequence with a value determined by its new state.

Note that it suffices to consider schedules where all accesses to each snapshot object
occur before any accesses to the next snapshot object in the sequence.
This is because if process $p_j$ performs its $\mathsf{update}$ of a particular snapshot
object after process $p_i$ performs its $\mathsf{scan}$s to this object, then it is indistinguishable to both processes how much later this occurs.

A \emph{configuration} consists of the contents of each shared object and the state of each process.
From any configuration $C$, a \emph{scheduler} decides the order in which the processes take steps. 
The sequence 
of processes selected by the scheduler is called a \emph{schedule from} $C$. Given a finite schedule $\alpha$ from $C$, we use $C\alpha$ to denote the resulting configuration. 
An algorithm
is \emph{wait-free} if there is no infinite schedule from any initial configuration.

The \emph{non-uniform iterated immediate snapshot} (NIIS) model,
introduced by Hoest and Shavit~\cite{hoest2006toward},
is like a full-information $(n-1)$-NIS model, except that the scheduler is restricted in how it can schedule processes: 
It repeatedly selects a set of processes that are all poised to perform $\mathsf{update}$s on the same snapshot object. Each of the processes in the set performs its $\mathsf{update}$. Then, each of these processes performs one $\mathsf{scan}$ of this snapshot object.
Note that, since each process performs an $\mathsf{update}$ to a snapshot object before performing a $\mathsf{scan}$, the $\mathsf{scan}$ will return a vector containing at most $n-1$ components with value $-$.
Initially, the state of process $p_i$ consists of its identifier, $i$, and its input.
When it performs an $\mathsf{update}$, the value it uses is
its current state.  
After performing a $\mathsf{scan}$, its new state consists of $i$ and the result of the $\mathsf{scan}$. 

Each initial configuration in the NIIS model or $f$-NIS model corresponds to a \emph{simplex} (or an $n$-vertex clique) containing one vertex for each process, which specifies its input. 
The collection of all such simplexes is called the \emph{input complex} (or \emph{input graph}).
Likewise, for any algorithm, each reachable terminal configuration
corresponds to a simplex (or $n$-vertex clique) containing one vertex for each process, which specifies its state,
including the value it outputs.
The collection of all such simplexes (or $n$-vertex cliques) is called
the \emph{protocol complex} (or \emph{protocol graph}).
We may assume that the sets of possible states of different processes are disjoint.
There is an edge between two vertices 
if they represent the states of different processes and there is a reachable configuration containing both these states.

A nice feature of the NIIS model is that the protocol complex of any wait-free algorithm 
can be obtained from the input complex by performing a finite number of \emph{non-uniform chromatic subdivisions} of simplexes.
In the special case when there are $n=3$ processes,
each simplex is a triangle and the
non-uniform chromatic subdivision of a simplex is
a triangulation of the simplex.
Likewise, in the $(n-1)$-NIS model, Alistarh, Aspnes, Ellen, Gelashvili, and Zhu~\cite{AAEGZ20}
have shown that the protocol graph of any wait-free, full-information algorithm can be obtained from the input graph by performing a finite number of \emph{subdivisions}
of $n$-vertex cliques.
For deterministic, wait-free computation, both the NIIS model and the $(n-1)$-NIS model 
are equivalent to the asynchronous shared memory model in which processes communicate using shared registers (which support only $\mathsf{read}$ and $\mathsf{write}$)~\cite{BG97}.

\subsection{The synchronous message-passing model}

In the synchronous message-passing model,
there is no uncertainty regarding the relative speeds of processes. A computation is divided into synchronous rounds. In each round, each process sends its entire state to every other process. Any message sent by a non-faulty process in round~$r$ is guaranteed to arrive at its destination before the end of round~$r$.
A synchronous algorithm is
an $f$-resilient solution to a task using $T$ rounds 
if all non-crashed processes decide on an output value by the end of round $T$
in any execution with at most $f$ crashes.

\section{Impossibility of asynchronous wait-free cycle agreement}
\label{sec:impossibility-wait-free}

In this section, we give a proof of the following result. 

\begin{theorem}\label{thm:no-wait-free}
For $c \ge 4$,
there is no wait-free algorithm for the $c$-cycle agreement problem
among $n \geq 3$ processes in the
NIIS
model.
\end{theorem}

Our proof relies on a slight generalisation of Sperner's lemma to convex polygons, originally shown by Atanassov~\cite{atanassov1996sperner} and generalised to convex polytopes of any dimension by de Loera, Peterson, and Su~\cite{deloera2002polytopal}. However, for us, a special case in the two-dimensional setting suffices.

\begin{figure}[t]
  \centering
  \includegraphics[page=2,width=0.95\textwidth]{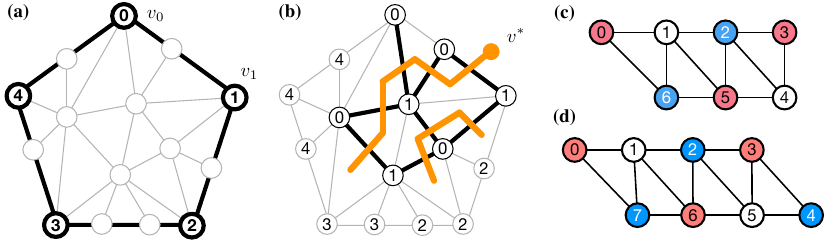}
\caption{(a) A triangulation $T$ of a pentagon. (b) A Sperner labelling of $T$ with the edges of the graph $G'$ superimposed in orange. (c) The subcomplex $\mathbb{H}$ for $c=7$. (d) The subcomplex $\mathbb{H}$~for~$c=8$. \label{fig:triangulations-and-complexes}}
\end{figure}

Let $H$ be a 
polygon with $c$ vertices and let $T$ be a triangulation of $H$.
A \emph{Sperner labelling} of $T$ is a function
from the vertices of $T$ to the set
$\{ 0, \ldots, {c-1} \}$ such that each vertex of $H$ gets a different label and each vertex on the boundary of $T$ between two vertices of $H$ gets the same label as one of those two vertices. 
Please see Figures~\ref{fig:triangulations-and-complexes}(a) and \ref{fig:triangulations-and-complexes}(b).

\begin{lemma}
Let $H$ be a convex polygon with $c$ vertices. Any Sperner labelling of a triangulation of $H$
has a triangle whose vertices have three different labels.
\label{lem:Sperner}
\end{lemma}
\begin{proof}
Let $T$ be a triangulation of $H$ and 
consider
any Sperner labelling of $T$.
Without loss of generality, suppose there are two adjacent vertices $v_0$ and $v_1$ of $H$ labelled with 0 and 1, respectively.
Consider the graph $G'=(V'\cup \{v^*\},E')$, where $V'$ is the set of triangles of $T$.
There is an edge in $E'$ between  triangles $\tau$ and $\tau'$ if and only if they have exactly
two vertices in common, one of which is labelled 0 and the other of which is labelled 1. 
There is an edge in $E'$ between $v^*$ and triangle $\tau$ if and only if two of the vertices
of $\tau$ lie on the boundary of $T$ between 
$v_0$ and $v_1$
and they have different labels.
This is illustrated in \figureref{fig:triangulations-and-complexes}(b).

Each of the nodes of $T$ on the boundary between $v_0$ and $v_1$ is labelled by 0 or 1.
The labels of the nodes on this path change an odd number of times,
since $v_0$ and $v_1$ have different labels.
Thus, there are an odd number of edges on the boundary whose endpoints are labelled 0 and 1,
so $v^*$ has odd degree.
If a triangle has two nodes labelled 0 and one node labelled 1 or vice versa, it has degree 2 in $G'$.
If a triangle has one node labelled 0, one node labelled 1, and one node with some other label, it has degree 1 in $G'$. Otherwise, it has degree~0~in~$G'$.

The handshaking lemma~\cite{euler1741solutio} says that any finite graph contains an even number of nodes with odd degree. Since $v^*$ has odd degree, there exists a triangle $\tau \in V$ with odd degree.
The vertices of this triangle have three different labels.
\end{proof}

\begin{proof}[Proof of \theoremref{thm:no-wait-free}.]
  Let $\mathbb H$ denote the part of the input complex for $c$-cycle agreement among 3 processes
  $p_0$, $p_1$, and $p_2$,
  consisting of the simplexes corresponding to the following $c-2$ input configurations:
  \begin{itemize}
\item
$p_0$ has input $3a$, $p_1$ has input $3a+1$, and $p_2$ has input $c-3a-1$, for $0 \leq a \leq\lfloor (c-3)/6 \rfloor$.
\item
$p_1$ has input $3a+1$, $p_2$ has input $c-3a - 1$, and $p_0$ has input $c-3a-2$, for $0 \leq a \leq \lfloor (c-4)/6 \rfloor$.
\item
$p_1$ has input $3a+1$, $p_2$ has input $3a+2$, and $p_0$ has input $c-3a-2$, for $0 \leq a \leq \lfloor (c-5)/6 \rfloor$.
\item
$p_2$ has input $3a+2$, $p_0$ has input $c-3a-2$, and $p_1$ has input $c-3a-3$, for $0 \leq a \leq \lfloor (c-6)/6 \rfloor$.
\item
  $p_2$ has input $3a+2$, $p_0$ has input $3a+3$, and $p_1$ has input $c-3a-3$, for $0 \leq a \leq \lfloor (c-7)/6 \rfloor$.
\item
$p_0$ has input $3a+3$, $p_1$ has input $c-3a - 3$, and $p_2$ has input $c-3a-4$, for $0 \leq a \leq \lfloor (c-8)/6 \rfloor$.
\end{itemize}
Note that $\mathbb H$ contains $c$ vertices, one for each possible input value.
The cases $c = 7$ and $c=8$ are illustrated in Figures~\ref{fig:triangulations-and-complexes}(c) and \ref{fig:triangulations-and-complexes}(d).
The processes $p_0$, $p_1$, and $p_2$
are denoted by the colours red, white, and blue, respectively. 
The border of $\mathbb H$ is a polygon $H$ with $c$ vertices.

Consider any wait-free algorithm for 3 processes in the NIIS model.
Let $\mathbb S$ denote its protocol complex. It is finite, since the algorithm is wait-free.
Let $\mathbb T$ denote the subcomplex of $\mathbb S$ consisting of all terminal configurations reachable from configurations in $\mathbb H$, where each vertex is labelled with
the output value it contains.
The vertices and edges of $\mathbb T$ form a triangulation $T$ of $H$.
For each 
input value $x \in \{0, \ldots, c-1\}$,
there is a vertex $v_x$ on the boundary of $\mathbb T$ that corresponds to the solo execution by 
some
process $p_i$ with input $x$. If it is not labelled by the value $x$,
then the algorithm does not solve $c$-cycle agreement.
The edges on the border of $\mathbb T$ between $v_x$ and $v_{x'}$,
where $x' = (x +1)\bmod c$, correspond to executions by only two processes, one with input $x$ and the other with input $x'$. If the endpoints of all such edges are not labelled by $x$ or $x'$, the algorithm does not solve $c$-cycle agreement.
Label each vertex of $T$ with the label of the corresponding vertex in $\mathbb T$.
If the algorithm is correct, then this is a Sperner labelling.
By Lemma \ref{lem:Sperner},
the triangulation 
$T$ contains a triangle whose vertices have three different labels.
The corresponding configuration is the result of an execution in which the three processes output different values,
so the algorithm cannot be solving $c$-cycle agreement among three processes.

Since all but three processes can crash before taking any steps, any algorithm that solves
$c$-cycle agreement among $n \geq 3$ processes is also an algorithm that solves $c$-cycle agreement among 3 processes. Therefore no such algorithm exists.
\end{proof}

\section{Impossibility results via reductions}
\label{sec:reductions}

In this section, we show that the impossibility of wait-free cycle agreement implies the impossibility of 2-resilient cycle agreement in the asynchronous shared memory model
(where processes communicate by reading from and writing to registers) and a lower bound on the round complexity of cycle agreement in the synchronous message model.
Finally, we show that approximate agreement is impossible on graphs that admit a certain 
labelling
of the vertices.

\subsection{There exists no 2-resilient asynchronous algorithm}

A task is {\em colourless} if the input of any process may be the input of any other process, the output of any process may be the output of any other process, and the specifications of valid outputs only depend on the set of inputs of the processes. Cycle agreement is an example of a colourless task.
The BG simulation technique~\cite{borowsky2001bg} shows that the impossibility of wait-free algorithms
for a colourless task
for $n \ge 3$ processes implies the impossibility of 2-resilient algorithms
for that task.

\begin{theorem}{\rm \cite{borowsky2001bg}}
  If there exists a $k$-resilient asynchronous algorithm for $n > k$ processes that solves a colourless task,
then there is a wait-free asynchronous algorithm for $(k+1)$ processes that solves the task.
\end{theorem}

Together with \theoremref{thm:no-wait-free}, the BG simulation immediately implies that there is no 2-resilient asynchronous algorithm for the cycle agreement problem.

\begin{corollary}\label{corollary:2-resilient}
  For any $n \ge 3$ and $c \ge 4$, there is no 2-resilient asynchronous algorithm that solves $c$-cycle agreement.
\end{corollary}

\subsection{Time lower bounds for synchronous algorithms}

We can now lift the impossibility results 
to time lower bounds for the synchronous model using the round-by-round simulation by Gafni~\cite{gafni1998round}, who showed the following. 

\begin{theorem}{\rm \cite{gafni1998round}}  Let $0 < k < f < n$ such that $n-k-f>0$. Fix $T \le f / k $. Suppose there exists a synchronous $f$-resilient algorithm for $n$ nodes that solves a colourless task 
in $T$ rounds. Then there exists a $k$-resilient asynchronous algorithm that solves the task.
\label{RoundByRound}
\end{theorem}

\noindent 
Applying \corollaryref{corollary:2-resilient} and
\theoremref{RoundByRound},
we obtain a time lower bound for synchronous algorithms.

\begin{corollary}\label{corollary:synchronous-lb}
 For any $n > f \ge 0$, any $f$-resilient synchronous message-passing algorithm for $c$-cycle agreement requires at least $\lfloor f/2 \rfloor + 1$
 rounds.
\end{corollary}

\subsection{Graphs on which approximate agreement is impossible}

We now show that approximate agreement is hard on graphs that admit a certain 
labelling of its vertices.
We do so by a reduction from 2-set agreement among $n \geq 3$ processes.
In this problem, each process has an input value in $\{0,1,2\}$
and, if it does not crash, it must output one of the inputs such that at most two different values are output.

A labelling $\ell \colon V \to \{0,1,2\}$ of the vertices of a graph $G = (V,E)$
is a \emph{lower bound labelling} if the following conditions hold:
\begin{enumerate}
    \item 
    $G$ contains no triangle with three different labels and
    \item 
    $G$ contains a cycle $C$ in which exactly one vertex has label 1 and its two neighbours in $C$ have 
    labels 0 and 2.
\end{enumerate}
It is easy to check that any cycle graph of length $c \ge 4$ admits
a lower bound labelling:
pick three consecutive vertices, label them with 0,1,2, and label all other vertices with 2. 
A wheel graph, which consists of a cycle and one central vertex that is a neighbour of all vertices in the cycle,
does not admit a lower bound labelling. 
On the other hand, if one edge adjacent to the central vertex is removed, the resulting graph does admit
a lower bound labelling: label the other endpoint of the removed edge with 1, label one of its neighbours with 0, and 
label all other vertices with 2.

\begin{theorem}
Suppose $G$ is a graph that admits a lower bound labelling. Then there is no wait-free algorithm among $n \ge 3$ processes that solves graphical approximate agreement on $G$.
\end{theorem}
\begin{proof}
Consider a lower bound labelling $\ell$ of $G$. Let $C$ be a cycle in $G$ that contains exactly one vertex, $v_1$, with label 1, a neighbour $v_0$ of $v_1$ with label 0, and
a neighbour $v_2$ of $v_1$ with label 2.
Let $A$ be a wait-free approximate agreement algorithm on the path $C\setminus\{ v_1 \}$.

To obtain a contradiction, suppose there is a wait-free algorithm $B$
for graphical approximate agreement on $G$. The following wait-free algorithm solves 2-set agreement:
\begin{itemize}
\item Processes with input value $x \in \{0,2\}$ run the approximate agreement algorithm $A$ 
on the path $C \setminus \{ v_1 \}$ 
using $v_x$ as input. The vertex each of these processes outputs in $A$ is used as its input for algorithm $B$.
\item Processes with input value 1 use $v_1$ as their input for algorithm $B$.
\item Each process $p_i$ outputs the label $\ell(y_i)$ of the vertex $y_i$ it outputs in $B$.
\end{itemize}
By the agreement property of graphical approximate agreement,
the values output in $B$ lie on a clique.
The first property of a lower bound labelling implies that
the nodes in this clique have at most two distinct labels. Thus, at most two different values are output by the processes.

If there are three distinct input values, then validity is immediately satisfied. If all input values are the same,
then all output values are this input value, since this is true for algorithms $A$ and $B$.
It remains to
consider instances of set agreement with exactly two input values. First, suppose the inputs for 
set agreement are in $\{0,1\}$. 
All processes with input 0 output $v_0$ in algorithm $A$,
since $v_0$ is the only value input to $A$.
Thus, 
each process uses either $v_0$ or $v_1$ as its input to
algorithm $B$. As $v_0$ and $v_1$ are adjacent in $G$,
each process outputs one of these two values in $B$,
by validity of graphical approximate agreement.
Hence, 
each process outputs a value 
in $\{\ell(v_0), \ell(v_1)\} = \{0,1\}$ for
set agreement, satisfying validity.
The case $\{1,2\}$ is symmetric.

Now, suppose that the inputs for set agreement are in $\{0,2\}$. Then each process uses either $v_0$ or $v_1$
as its input to algorithm $A$. Their outputs in $A$
and, hence their inputs to algorithm $B$,
all lie on some edge $\{u,v\}$ on the path $C \setminus \{ v_1 \}$. By validity of graphical approximate agreement, 
each process outputs either $u$ or $v$ in $B$.
From the second property of a lower bound labelling, all values in $C \setminus \{v_1\}$ are labelled with either 0 or 2. 
Thus, each process outputs 0 or 2 for set agreement,
satisfying validity.
\end{proof}

\section{Impossibility of extension-based proofs}\label{sec:extension}

 Extension-based proofs were introduced by Alistarh, Aspnes, Ellen, Gelashvili and Zhu~\cite{alistarh2019extension} to model inductive impossibility arguments, such as the valency-based impossibility of consensus in asynchronous message-passing systems by Fisher, Lynch and Paterson~\cite{fischer85impossibility}. These are in contrast to the combinatorial arguments used to show the impossibility of set agreement~\cite{borowsky1993generalized,herlihy1999topological,saks2000wait}. 
It is known that extension-based proofs cannot be used to prove the impossibility of $(n-1)$-set agreement among $n > 2$ processes in the NIIS model~\cite{alistarh2019extension} or in the $(n-1)$-NIS model~\cite{AAEGZ20}.

We show that extension-based proofs cannot be used to prove Theorem~\ref{thm:no-wait-free} in the $(n-1)$-NIS model. 
This is the first application of the extension-based proof framework to a task other than set agreement. 
We emphasise that this result does not follow directly via reduction from the result for $k$-set agreement.
The main source of novelty in our argument is in carefully extending their adversarial protocol to the $c$-cycle agreement task. 
Specifically, our main result is the following. 

\begin{restatable}{theorem}{extension}
\label{thm:extension}
There is no extension-based proof of the impossibility of a wait-free algorithm solving 4-cycle agreement for $n \geq 3$ processes in the $(n-1)$-NIS model.
\end{restatable}

\subsection{Extension-based proofs}

We follow the terminology and notation given in~\cite{alistarh2019extension}, and the outline of our argument is similar.
However, care is needed to apply the argument to the $c$-cycle agreement problem, because of the differences in its specification.

An \emph{extension-based} proof is 
an interaction between a prover and any full-information protocol.
The prover starts with no knowledge about the protocol (except its initial configurations) and makes the protocol reveal information about the states of processes in various configurations by asking \emph{queries}.
The interaction proceeds in phases.

In each phase $\varphi \geq 1$, the prover starts with a finite schedule,
$\alpha(\varphi)$, and a set, $\config{A}(\varphi)$, of configurations that are reached by performing
$\alpha(\varphi)$ from initial configurations. These initial configurations only differ from one another in the input values of
processes that do not appear in the schedule $\alpha(\varphi)$.
If every configuration in $\config{A}(\varphi)$ is terminal and the outputs satisfy the specification of the task, then the prover \emph{loses}.

The prover also maintains a set, $\config{A}'(\varphi)$, containing the configurations it reaches 
by non-empty schedules from configurations in  $\config{A}(\varphi)$ during phase $\varphi$.
This set is empty at the start of phase $\varphi$. 
At the start of the first phase, indexed by $1$, $\alpha(1)$ is the empty schedule and $\config{A}(1)$ is the set of all initial configurations of the protocol.

The prover \emph{queries} the protocol
by specifying a configuration $C \in \config{A} (\varphi) \cup \config{A}'(\varphi)$ and 
a  process $q$ that is active (i.e.~has not terminated) in $C$.
Let $C'$ be the configuration resulting from scheduling one step of $q$ from $C$.
The protocol replies to this query with the state $s$ of $p$ in $C'$.
(For a full-information protocol specified by the function $\Delta$, it suffices for the protocol to reply with $\Delta(s)$.)
Then the prover adds $C'$ to $\config{A}'(\varphi)$
and we say that the prover has \emph{reached} $C'$.
If the prover reaches a configuration $C'$ in which the outputs of the processes do not satisfy
the specifications of the task, it has demonstrated that the protocol is incorrect.
In this case, the prover \emph{wins}.
A \emph{chain of queries} is a (finite or infinite) sequence of queries such that, for all consecutive queries $(C_i,q_i)$ and $(C_{i+1},q_{i+1})$  in the chain, $C_{i+1}$ is the configuration resulting from scheduling one step of $q_i$ from $C_i$.

An \emph{output} query in phase $\varphi$ is specified by a configuration $C \in \config{A}(\varphi) \cup \config{A}'(\varphi)$, a set of active processes $Q$ in $C$, 
and a possible output value $y$. 
If there is a schedule from $C$ 
involving only processes in $Q$ (i.e.~a \emph{$Q$-only} schedule) that results in a configuration in which some process in $Q$ outputs $y$,
then the protocol returns some such schedule. Otherwise, the protocol returns \none.
In this case, if the prover later reaches a configuration by a $Q$-only schedule starting from $C$ in which some process in $Q$ outputs $y$,
the protocol has responded inconsistently and
the prover wins.

After constructing finitely many output queries and chains of queries in phase $\varphi$ without winning, the prover must end the phase by committing to a non-empty extension $\alpha'$ of the schedule $\alpha(\varphi)$ such that $C\alpha' \in \config{A}'(\varphi)$ for some $C \in \config{A}(\varphi)$.
Since there is an initial configuration $C_0$ such that $C$ is reached by performing $\alpha(\varphi)$ starting from $C_0$, configuration $C\alpha'$ is reached by performing $\alpha(\varphi+1)=\alpha(\varphi) \alpha'$ starting from $I$.
The prover defines $\config{A}(\varphi+1)$ to be the set of all configurations
that are reached by performing $\alpha(\varphi+1)$ from the initial configurations that
only differ from $C_0$ by the states
of processes that do not appear in this schedule.
Then the prover begins phase $\varphi+1$.

If the interaction between the prover and the protocol is infinite, either because the prover 
constructs an infinite chain of queries
or the number of phases is infinite, the prover \emph{wins}.
In this case, the prover has demonstrated that the protocol is not wait-free.
To prove that a task is impossible using an extension-based proof, one must show there exists a prover that  wins against {\em every} protocol. 

Our main result in this section is the following. 

\begin{theorem}
There is no extension-based proof of the impossibility of a wait-free algorithm solving 4-cycle agreement for $n \geq 3$ processes in the $(n-1)$-NIS model.
\end{theorem}

\subsection{Preliminaries and invariants}

Let $\mathbb{G}_0$ denote the input graph, which is the union of all $n$-vertex cliques representing input configurations of a protocol.
For each $t \geq 1$, let  $\mathbb{G}_t$ denote the graph consisting of the union of all $n$-vertex cliques representing configurations of the protocol reachable from initial configurations by schedules in which each process performed a $\mathsf{scan}$ on $S_t$ during its last step
or terminated before accessing $S_t$.
Given $\Delta(v)$ for each vertex $v$ of an $n$-vertex clique $\sigma$ of $\mathbb{G}_{t-1}$, the subdivision,
$\chi(\sigma,\Delta)$ is the union of all $n$-vertex cliques representing configurations of
the protocol reachable from the configuration represented by $\sigma$ via schedules in which each active
process takes exactly two steps.
The subdivision of any union of $n$-vertex cliques is the union of the subdivisions of those cliques.
In particular,  $\mathbb{G}_t = \chi(\mathbb{G}_{t-1},\Delta)$.
A vertex is \emph{terminated} if it represents the state of a process that has terminated.
Otherwise, the vertex is \emph{active}.
If $\mathbb{T}$ is a set of terminated vertices in $\mathbb{G}_{t-1}$,
then $\chi(\mathbb{T},\Delta) = \mathbb{T}$ is a set of terminated vertices in $\mathbb{G}_t$.

We will use the following result,
from \cite{AAEGZ20}, which relates the distance between two sets of vertices in $\mathbb{G}_{t-1}$ to the distance between their subdivisions in $\mathbb{G}_{t}$.

\begin{lemma}[The Distance Lemma]
	\label{lem:dist}
Suppose $\mathbb{A}$ and $\mathbb{B}$ are non-empty and each is either 
the union of $n$-vertex cliques in $\mathbb{G}_{t-1}$ or a set of terminated vertices in $\mathbb{G}_{t-1}$.
Then the distance between $\chi(\mathbb{A},\Delta)$ and  $\chi(\mathbb{B},\Delta)$ in $\mathbb{G}_t$
is at least as large as the distance between $\mathbb{A}$ and $\mathbb{B}$ in $\mathbb{G}_{t-1}$.
Moreover,
if every path between $\mathbb{A}$ and $\mathbb{B}$ in $\mathbb{G}_{t-1}$ contains at least one edge between active vertices, then the distance between $\chi(\mathbb{A},\Delta)$ and  $\chi(\mathbb{B},\Delta)$ in $\mathbb{G}_t$
is larger than the distance between $\mathbb{A}$ and $\mathbb{B}$ in $\mathbb{G}_{t-1}$.
\end{lemma}

We define an adversary that is able to win against every extension-based prover, which is attempting to prove the impossibility of $4$-cycle agreement for $n \geq 3$ processes. The adversary maintains a \emph{partial} specification of $\Delta$ (the protocol it is adaptively constructing) and an integer $t \geq 0$. The integer $t$ represents the number of subdivisions of the input graph, $\mathbb{G}_0$, that it has performed.  
Once the adversary has defined $\Delta$ for each vertex in $\mathbb{G}_{t-1}$, it may subdivide $\mathbb{G}_{t-1}$ and construct $\mathbb{G}_t = \chi(\mathbb{G}_{t-1},\Delta)$. 

Let $0 \leq r \leq t$ and let  $a \in \{0,1,2,3\}$ be an input value.
Throughout 
this section, addition and subtraction on input values are always taken modulo 4.
We say that a vertex \emph{has seen} a value $a$ if it represents the state of a process which has
seen $a$ in some $\mathsf{scan}$ (i.e.~the process has $a$ as its input or has seen the
$\mathsf{update}$ of a process that had previously seen $a$).
The following definitions are the key to the proof. In particular, $\Rt{a,a+1}$ replaces $\Nt{a}$, which
was used in the proof for $k$-set agreement.
\begin{itemize}
\item
$\Rt{a,a+1}$ is the subgraph of $\Sr$ consisting of the union of all $n$-vertex cliques in $\Sr$ whose vertices have only seen values $a$ or $a+1$.
\item
$\Tt{a}$ is the set of terminated vertices in $\Sr$ that have output value $a$, i.e.~vertices for which $\Delta$ is $a$.
\item
$\Xt{a}$ is the set of vertices in $\Sr$ that represent the states of processes in $Q$ in configurations reachable from $C$ by a $Q$-only schedule, for some output query $(C,Q,a)$ to which the adversary answered \none.
\end{itemize}

The next result follows from the definition of subdivision.

\begin{proposition}
\label{prop:subdiv}
For $t \geq 0$ and all inputs $a$,
$\Rt{a,a+1}$ is non-empty
and $\Rto{a,a+1} = \chi(\Rt{a,a+1}, \Delta)$.
If $\Tt{a}$ is non-empty, then $\chi(\Tt{a}, \Delta) = \Tt{a}$.
\end{proposition}

\subsection{The adversarial strategy}

Our adversarial strategy ensures that after each response to a query made by the prover in phase~1, the following invariants will hold:
\begin{enumerate}
\item \label{inv:definedr}
For each $0 \leq r < t$ and each vertex $v \in \Sr$, $\Delta(v)$ is defined. 
\item \label{inv:definedt}
If $v$ is a vertex in $\St$, then 
$\Delta(v) \neq \bot$.
\item \label{inv:defineds}
If $s$ represents the state of a process in a configuration that was reached by the prover and the process took $2r$ steps  
in the execution to reach this configuration, then $s$ is a vertex in $\Sr$, for some $0 \leq r \leq t$, and $\Delta(s)$ is defined. 
\item \label{inv:outputsfar}
For any two inputs $a \neq b$, if $\Tt{a}$ and $\Tt{b}$ are non-empty, then 
the distance between them in $\St$ is at least 3.
\item \label{inv:contains}
For any input $a$, if $\Tt{a}$ is non-empty
and $a\neq b,b+1$, then the distance between $\Tt{a}$ and $\Rt{b,b+1}$ is at least 2.
\item \label{inv:none}
For every input $a$, every vertex in $\Xt{a}$ is either in 
$\Rt{b,b+1}$, where $a \neq b,b+1$,
or at distance at most one from
$\Tt{b}$, for some $b \neq a$.
\end{enumerate}

\noindent The following lemma is a consequence of the invariants. 

\begin{lemma}
\label{lem:active-edge}
If $a \neq b$,
then any path between $\Tt{a}$ and $\Tt{b}$ contains an edge between active vertices.
If $a \neq b,b+1$, then any path between $\Tt{a}$ and $\Rt{b,b+1}$ 
contains an edge between active vertices.
\end{lemma}

\begin{proof}
Consider any path $v_0,v_1,\dots,v_\ell$ between $\Tt{a}$ and $\Tt{b} \cup \Rt{b,b+1}$ 
in $\St$.
Let $v_j$ be the last vertex in $\Tt{a}$.
Then, by invariants \ref{inv:outputsfar} and \ref{inv:contains}, $\ell \geq j+2$.
Since $v_j$ is the last vertex in $\mathbb{T}_a^{t}$, $v_{j+1}$ and $v_{j+2}$ are not in $\mathbb{T}_a^{t}$. Moreover, by invariant \ref{inv:outputsfar}, $v_{j+1}$ and $v_{j+2}$ are not in $\mathbb{T}_c^{t}$ for any input $c \neq a$.
Hence, $\{v_{j+1},v_{j+2}\}$ is an edge between active vertices. 
\end{proof}

\noindent
A subdivision maintains the invariants, but increases the distance between vertices that output different values
and between vertices that output a value and vertices that have only seen a different value or two adjacent different values.

\begin{lemma}
\label{lem:invariants-hold}
Suppose all the invariants hold, the adversary defines $\Delta(v) = \bot$ for each vertex $v$ in $\St$
where $\Delta$ is undefined and $\St$ to construct $\Sto$.
If $a \neq b$ and both $\Tt{a}$ and $\Tt{b}$ are non-empty, then the distance between $\Tto{a}$ and $\Tto{b}$ in $\Sto$ is greater than the distance between $\Tt{a}$ and $\Tt{b}$ in $\St$.
If $a \neq b,b+1$ and $\Tt{a}$ is non-empty, then the distance between $\Tto{a}$ and $\Rto{b,b+1}$ in $\Sto$
is greater than the distance between $\Tt{a}$ and $\Rt{b,b+1}$ in $\St$.
Furthermore, if the adversary increments $t$, then all the invariants continue to hold.
\end{lemma}

\begin{proof}
Suppose that $\Tt{a}$ is non-empty.
By Lemma \ref{lem:active-edge}, any path between $\Tt{a}$ and $\Rt{b,b+1}$ for $a\neq b, b+1$ contains an edge between active vertices.
Hence, by Lemma \ref{lem:dist}, the distance between
$\Tto{a} = \chi(\Tt{a},\Delta)$ and $\Rto{b,b+1} = \chi(\Rt{b,b+1},\Delta)$ in $\Sto$ is larger than the distance between $\Tt{a}$ and $\Rt{b,b+1}$ in $\St$.
Similarly, if $a\neq b$ and $\Tt{b}$ is non-empty, then the distance between
$\Tto{a} = \chi(\Tt{a},\Delta)$ and $\Tto{b} = \chi(\Tt{b},\Delta)$ in $\Sto$ is larger than the distance between $\Tt{a}$ and $\Tt{b}$ in $\St$.
Hence, invariants \ref{inv:outputsfar} and \ref{inv:contains} remain true after $t$ is incremented.

Before incrementing $t$, the adversary defines $\Delta(v) = \bot$ for each vertex $v \in \St$ where $\Delta(v)$ was undefined. Since invariant \ref{inv:definedr} was true, it remains true.
Invariant \ref{inv:definedt} and \ref{inv:none} hold by construction and the definition of subdivision.
Invariant \ref{inv:defineds} is not affected.
\end{proof}

\medskip
\paragraph{The adversarial strategy for phase 1.}
Initially, 
the adversary sets $\Delta(v) = \bot$ for each vertex $v \in \mathbb{S}_0$. It then subdivides $\mathbb{S}_0$ to construct $\mathbb{S}_1$ and sets $t=1$. 
This ensures that invariants \ref{inv:definedr} and
\ref{inv:definedt} are true.
Invariant \ref{inv:defineds} is true because,
before the first query, the prover has only reached initial configurations and $\mathbb{S}_0$ is the union of all $n$-vertex cliques representing initial configurations.
No vertices in $\mathbb{S}_0$ have terminated, so $\mathbb{T}_0(a)$ is empty for each input $a$.
Thus invariants \ref{inv:outputsfar} and \ref{inv:contains} are vacuously true.
No output queries have been performed, so $\mathbb{X}_0(a)$ is also empty for each input $a$
and invariant \ref{inv:none} is vacuously true.
So, suppose that the invariants are satisfied immediately prior to some query by the prover during phase 1.

First, consider a query $(C,q)$, where $C$ is a configuration previously reached by the prover and $q$ is an active process in $C$.
If $q$ took $2r$ steps in the execution to reach $C$,
then, by invariant \ref{inv:defineds}, the state $s$ of process $q$ in configuration $C$ is a vertex
in $\Sr$ and $\Delta(s)$ is defined.
Since $q$ is active in $C$, $\Delta(s) = \bot$, so, by invariant \ref{inv:definedt}, $r < t$.
In this case, the adversary returns the configuration $Cq$, which is the same as $C$, except that the component belonging to $q$ in $S_{r+1}$
changes value from $-$ to $s$
and process $q$ has now performed $2r+1$ steps.
Process $q$ remains active.
Invariant \ref{inv:defineds} remains true: it holds vacuously for the new state of process $q$ and no other process has changed state.
Since $\Delta$ has not been changed, $\Tt{a}$ has not changed for any input $a$ and invariants \ref{inv:definedr}, \ref{inv:definedt},\ref{inv:outputsfar} and \ref{inv:contains} remain true.
Since no vertices are added to $\Xt{a}$ for any input a, invariant \ref{inv:none} remains true.

So, suppose that $q$ took $2r+1$ steps in the execution to reach $C$. Let $s$ be the previous state
of process $q$ in this execution 
and 
consider
the last configuration in this execution in which $q$ had state $s$. Then, by invariant \ref{inv:definedt},
$s$ is a vertex in $\Sr$. Since $q$ is active
in this configuration, $\Delta(s) = \bot$.
Hence, by invariant \ref{inv:definedr}, $r <t$.
The state $s'$ of $q$ in configuration $Cq$ consists of its id
and the result of its $\mathsf{scan}$ of $S_{r+1}$.
It is a vertex in $\Sro$.

If $r < t-1$, then, by invariant \ref{inv:definedr}, $\Delta(s')$ is defined. It is also possible that $r=t-1$ and $\Delta(s')$ is defined.
In both these cases, 
all invariants continue to hold.

Now suppose that $r = t-1$ and $\Delta(s')$ is not defined.
Suppose there exists an input $a$ such that setting $\Delta(s') = a$ maintains all the invariants.
Specifically, suppose the following properties hold: 
\begin{itemize}
\item
for all inputs $b\neq a$ such that $\Tt{b}$ is non-empty, the distance between $s'$ and $\Tt{b}$ in $\St$ is at least 3 and 
\item
for all inputs $b$ such that $a \neq b,b+1$,
the distance between $s'$ and $\Rt{b,b+1}$ in $\St$ is at least 2.
\end{itemize}
In this case, the adversary defines $\Delta(s') = a$. This adds the vertex $s'$ to $\Tt{a}$ and leaves
$\Tt{b}$ unchanged for $b \neq a$. It also does not change $\Rt{b,b+1}$ or $\Xt{b}$ for any input $b$.
Hence invariants \ref{inv:definedr},~\ref{inv:definedt},~\ref{inv:outputsfar},~\ref{inv:contains}, and~\ref{inv:none}
continue to hold. 
By construction, $s' \in \St$ and $\Delta(s')$ is defined. For every other process, its state in $Cq$
is the same as its state in $C$. Thus invariant \ref{inv:defineds} continues to hold.
By invariant \ref{inv:none},
each vertex $u \in \Xt{a}$ is either in $\Rt{b,b+1}$, for some value $b$ such that $a \neq b,b+1$,
or is distance at most one from $\Tt{b}$ for some $b \neq a$.
Since the distance  between $s'$ and $\Tt{b}$ in $\St$ is at least 3,
the distance between $s'$ and $u$ is at least 2. Thus $s' \not\in \Xt{a}$, so defining $\Delta(s') = a$
does not contradict the result of any previous output query.
Otherwise, 
the adversary defines $\Delta(v) = \bot$ for each vertex $v \in \St$ where $\Delta(v)$ is undefined,
including $s'$, subdivides $\St$ to construct $\Sto$, and increments $t$.
By Lemma~\ref{lem:invariants-hold}, all the invariants continue to hold.
In all cases, the adversary returns $s'$ and $\Delta(s')$.

\medskip

Second, for an output query $(C,Q,y)$, let $\mathbb{Q}$ be the set of vertices in $\St$
vertices representing the states of processes in $Q$ in configurations reachable from $C$ by $Q$-only schedules.
If some vertex $v \in \mathbb{Q}$ has terminated with output $y$, then the adversary returns a $Q$-only  schedule from $C$ that leads to a configuration  in which $v$ represents the state of some process.
If every vertex in $\mathbb{Q}$ is in $\Rt{y+1,y+2} \cup \Rt{y+2,y+3} \cup \Xt{y}$ or has terminated with an output other than $y$,
then it would be impossible for the adversary to return a $Q$-only  schedule from $C$ in which some vertex has terminated with output $y$ without violating validity or contradicting one of its previous answers, so
the adversary returns \none.
Note that adding vertices in $\Rt{y+1,y+2} \cup \Rt{y+2,y+3} \cup \Tt{a}$ for $a \neq y$ does not make invariant
\ref{inv:none} false.
Invariants \ref{inv:definedr}, \ref{inv:definedt},\ref{inv:defineds}, \ref{inv:outputsfar}, and \ref{inv:contains} also continue to hold.

Otherwise, let $\mathbb{U} \neq \emptyset$ be the 
subset of vertices in $\mathbb{Q}$ that are not in
$\Rt{y+1,y+2} \cup \Rt{y+2,y+3}$, $\Xt{y}$, or $\Tt{a}$, for any $a \neq y$.
Note that, by invariant \ref{inv:definedt}, $\Delta(u)$ is undefined for all $u \in \mathbb{U}$.
For each vertex $u \in \mathbb{U}$, let $\mathbb{A}_u$ be the union of all $n$-vertex cliques in $\St$ containing $u$.
We consider three cases.

\textbf{Case 1}: \emph{There is a vertex $u \in  \mathbb{U}$ 
such that $\mathbb{A}_u \cap \Tt{y}$ is non-empty.}
Then the adversary defines $\Delta(v) = \bot$ for each vertex $v \in \St$ where $\Delta(v)$ is undefined and subdivides $\St$ to construct $\Sto$. 
By Lemma~\ref{lem:invariants-hold}, all the invariants continue to hold,
if $a \neq y$ and $\Tt{a}$ is non-empty, then
the distance between $\Tto{y}$ and $\Tto{a}$ in $\Sto$ is at least 4,
and if $y \neq b,b+1$, then  
the distance between $\Tto{y}$ and $\Rto{b,b+1}$ in $\Sto$ is at least 3.

Let $p_i$ be the process whose state is $u$.
Since $u \in \mathbb{U} \subseteq \mathbb{Q}$, process $p_i \in Q$. 
Let $w \in \mathbb{A}_u \cap \Tt{y}$, let $\sigma$ be an $n$-vertex clique in $\mathbb{A}_u$ that contains $w$, and
let $C'$ be the configuration represented by $\sigma$.
Let $v$ be the vertex corresponding to the state of process $p_i$ after 
it takes two steps starting from $C'$ (an $\mathsf{update}$
of $S_{t+1}$ followed by a $\mathsf{scan}$ of $S_{t+1}$).
Next, the adversary increments $t$. All the invariants continue to hold,
by Lemma~\ref{lem:invariants-hold}.
Finally, the adversary defines $\delta(v) = y$ and
returns a $Q$-only schedule from $C$ that results in process $p_i$  being in state $v$.
This adds vertex $v$ to $\Tt{y}$.
Invariants  \ref{inv:definedr}, \ref{inv:definedt}, \ref{inv:defineds}, and  \ref{inv:none} continue to hold.

Since $w$ is terminated, $w$ is adjacent to every vertex in $\chi(\sigma,\Delta) \subseteq \St$, including $v$.
It follows that, if $\Tt{a}$ is non-empty, then the distance
between $v$ and $\Tt{a}$ in $\St$ is at least 3
and the distance
between $v$ and $\Rt{b,b+1}$ in $\St$ is at least 2.
Thus, invariants \ref{inv:outputsfar} and \ref{inv:contains} hold.
By invariant \ref{inv:none},
every vertex in $\Xt{a}$ is either in 
$\Rt{b,b+1}$, where $a \neq b,b+1$,
or at distance at most one from
$\Tt{b}$, for some $b \neq a$.
Since the distance between $v$ and $\Rt{b,b+1}$ in $\St$ is at least 2 and
the distance between $v$ and $\Tt{b}$ in $\St$ is at least 3,
the distance between $v$ and $\Xt{a}$ is at least 2. Thus $v \not\in \Xt{y}$, so defining $\Delta(v) = y$
does not contradict the result of any previous output query.

\textbf{Case 2}:  \emph{There is a vertex $u \in  \mathbb{U}$ such that no vertex in $\mathbb{A}_u$ is terminated.}
In this case, the adversary defines
$\delta(v) = \bot$ for each vertex $v \in \St$ where $\Delta(v)$ is undefined and subdivides $\St$ to construct $\Sto$. 

Since no vertex in $\mathbb{A}_u$ is terminated and $\mathbb{A}_u$ contains all vertices at distance at most 1 from
$u$ in $\St$, the distance from $u$ to $\Tt{a}$ in $\St$ is at least 2 for all inputs $a$. 
Moreover, since $u \not\in \Rt{y+1,y+2} \cup \Rt{y+2,y+3}$, the distance from $u$ to $\Rt{y+1,y+2} \cup \Rt{y+2,y+3}$ in $\St$ is at least 1.

Let $p_i$ be the process whose state is $u$.
Since $u \in \mathbb{U} \subseteq \mathbb{Q}$, process $p_i \in Q$.
In the configuration represented by any $n$-vertex clique in $\St$, all components of $S_{t+1}$ are $-$.
Thus, the state of $p_i$ after it takes 2 steps starting from any configuration
in $\St$ that contains $u$ is the same. Let $v$ be the vertex in $\Sto$ that 
represents this state.

Consider any vertex $v'$ adjacent to $v$ in $\Sto$ and let $p_j$ be the process whose state is $v'$.
Then there exists 
a configuration in which $p_i$ is in state $v$ and $p_j$ is in state $v'$.
In state $v$, $p_i$ has not seen the $\mathsf{update}$ to $S_{t+1}$ by $p_j$,
so, in state $v'$, $p_j$ has not seen the $\mathsf{update}$ to $S_{t+1}$ by $p_i$. Since $u \not\in \Rt{y+1,y+2} \cup \Rt{y+2,y+3}$,
$v \not\in \Rto{y+1,y+2} \cup \Rto{y+2,y+3}$ and, hence,
$v' \not\in \Rto{y+1,y+2} \cup \Rto{y+2,y+3}$.
Thus the distance from $v$ to $\Rto{y+1,y+2} \cup \Rto{y+2,y+3}$ in  $\Sto$
is at least 2.

Consider any vertex $v''$ adjacent to $v'$ in $\Sto$. Then there exists an $n$-vertex clique $\sigma'$ such that $v',v'' \in \chi(\sigma',\Delta)$.
Since $v' \in \chi(\sigma',\Delta)$, the definition of subdivision implies that $u \in \sigma'$. Hence $\sigma' \subseteq \mathbb{A}_u$ and $v',v'' \in \chi(\mathbb{A}_u,\Delta)$.
Since no vertex in $\mathbb{A}_u$ is terminated, the distance between
$\mathbb{A}_u$ and $\Tt{a}
$ is at least 1 for every input $a$.
By Lemma~\ref{lem:dist}, the distance between $\chi(\mathbb{A}_u,\Delta)$
and $\chi(\Tt{a},\Delta)$ is at least 1.
By Proposition~\ref{prop:subdiv}, $\Tto{a} = \chi(\Tt{a},\Delta)$.
Thus, $v',v'' \not\in \Tto{a}$.
This implies that the distance from $v$ to $\Tto{a}$ in $\Sto$ is at least 3.

Now the adversary increments $t$, so all the invariants continue to hold,
by Lemma~\ref{lem:invariants-hold}. Finally,
the adversary defines $\Delta(v) = y$ and returns a $Q$-only schedule from $C$
that results in process $p_i$ being in state $v$.
This adds vertex $v$ to $\Tt{y}$.
Invariants  \ref{inv:definedr}, \ref{inv:definedt}, \ref{inv:defineds}, and  \ref{inv:none} continue to hold.
Since 
the distance from $v$ to $\Tt{a}$ is at least 3
and the distance from $v$ to $\Rt{y+1,y+2} \cup \Rt{y+2,y+3}$ in $\St$ is at least 2,
invariants \ref{inv:outputsfar} and \ref{inv:contains} hold.
As in the previous case,
defining $\Delta(v) = y$ does not contradict the result of any previous output query.

\textbf{Case 3.} \emph{For every simplex $\sigma \subseteq \mathbb{U}$, some vertex $w \in \mathbb{A}_\sigma$ has terminated with an output other than $y$}.
In this case, the adversary returns $\none$ and adds $\mathbb{U}$ to $\Xt{y}$.
Since each vertex in $\mathbb{U}$ is adjacent to some vertex that has terminated with an output other than $y$, invariant \ref{inv:none} holds.
Since $t$ and $\Delta$ are not changed, invariants  
\ref{inv:definedr}, \ref{inv:definedt}, and \ref{inv:defineds} continue to hold.
Since $\Tt{a}$ and $\Rt{a,a+1}$ are not changed for any input $a$,  
invariants \ref{inv:outputsfar} and \ref{inv:contains} still hold.

\subsection{The prover does not win in phase 1} 
Suppose that the invariants all hold before and after each query made by the prover in phase 1. By invariant~\ref{inv:outputsfar}, at most one value is output in any configuration reached by the prover. Moreover, by invariant~\ref{inv:contains}, if a process outputs a value $a$, then there does not exist $b$ such that $a \neq b,b+1$ and all the inputs are endpoints of the
edge $\{b, b+1\}$. Hence, the prover cannot win in phase~1 by showing that the protocol violates agreement or validity. It remains to  show that the prover cannot win by constructing an infinite chain of queries in phase 1.

\begin{lemma}
\label{lem:finitechains}
Every chain of queries in phase 1 is finite.
\end{lemma}
\begin{proof}
Assume, for a contradiction, that there is an infinite chain of queries, $(C_j,q_j)$, for $j \geq 0$.
Let $P$ be the set of processes that are scheduled infinitely often in this chain. 
Then there exists $j_0 \geq 0$ such that, for all $j \geq j_0$, $q_j \in P$. 
Let $t_0 \geq 1$ be the value of $t$ held by the adversary immediately prior to query $(C_{j_0},q_{j_0})$. 
By invariant~\ref{inv:defineds},  every process has taken fewer than $2t_0+2$ steps in the
schedule to reach configuration $C_{j_0}$, so no process has accessed $S_r$ in this execution,
for all $r \geq t_0+1$.
Thus, during the chain of queries, only processes in $P$ access $S_r$ for $r \geq t_0+1$.
Since all the processes in $P$ eventually are scheduled infinitely often in this chain,
the adversary eventually defines $\Delta(v) = \bot$
for each vertex $v \in \Sr$ where $\Delta(v)$ is undefined
subdivides $\Sr$ to construct $\mathbb{G}_{r+1}$, and increments $t$ from $r$ to $r+1$, for all $r \geq t_0$.

Since no process in $P$ ever terminates, $\Tr{a} = \mathbb{T}_{t_0}{a}$, for all inputs $a$ and all $r > t_0$. 
By invariant~\ref{inv:outputsfar},  if $\mathbb{T}_{t_0}{a}$ and $\mathbb{T}_{t_0}{b}$ are non-empty
and $a\neq b$,  the distance between
$\mathbb{T}_{t_0}(a)$ and $\mathbb{T}_{t_0}(b)$ in $\mathbb{G}_{t_0}$ is at least 3
and, so, by Lemma~\ref{lem:invariants-hold}, the distance between
$\mathbb{T}_{t_0+2}(a)$ and $\mathbb{T}_{t_0}{b}$ in $\mathbb{G}_{t_0+2}$ is at least 5.
Similarly, by invariant \ref{inv:contains} and Lemma~\ref{lem:invariants-hold},
if $\mathbb{T}_{t_0}(a)$ is non-empty, the distance between
$\mathbb{T}_{t_0+2}(a)$ and $\mathbb{R}_{t_0+2}(a+1,a+2) \cup \mathbb{R}_{t_0+2}(a+2,a+3)$
in $\mathbb{G}_{t_0+2}$ is at least 4.

Consider the first $j_1 \geq j_0$ such that process $q_{j_1}$ is poised to $\mathsf{scan}$ the snapshot object $S_{t_0+2}$ in $C_{j_1}$. By invariant~\ref{inv:defineds}, the state of process $q_{j_1}$ in configuration
$C_{j_1+1}= C_{j_1}q_{j_1}$ is a vertex $v$ in $\mathbb{G}_{t_0+2}$.
If there is some input $a$ such that $\mathbb{T}_{t_0+2}(a)$ is non-empty and
the distance from $v$ to $\mathbb{T}_{t_0+2}(a)$ in $\mathbb{G}_{t_0+2}$
is at most 2.
Then the distance from $v$ to $\mathbb{T}_{t_0+2}{b}$ in $\mathbb{G}_{t_0+2}$ is at least 3
for all $b \neq a$ such that $\mathbb{T}_{t_0+2}(b)$ is non-empty
and the distance from $v$ to
$\mathbb{R}_{t_0+2}(a+1,a+2) \cup \mathbb{R}_{t_0+2}(a+2,a+3)$ in $\mathbb{G}_{t_0+2}$ is at least 2.
According to its strategy for phase 1, the adversary defines $\Delta(v) = a$ after query $(C_{j_1},q_{j_1})$.
This contradicts the definition of $P$, since process $q_{j_1}$ terminates.
Thus, the distance from $v$ to any terminated vertex in $\mathbb{G}_{t_0+2}$ is at least 3.
Consider any $n$-vertex clique $\sigma$ in $\mathbb{G}_{t_0+2}$ that represents a configuration reachable from
configuration $C_{j_1+1}$. Since $v$ is a vertex in $\sigma$, the distance from $\sigma$
to any terminated vertex in $\mathbb{G}_{t_0+2}$ is at least 2. In particular, all vertices in $\sigma$ are
active.
Let $a$ be the input of process $q_{j_1}$. Since $q_{j_1}$ performed its $\mathsf{update}$ to $S_{t_0+2}$
prior to configuration $C_{j_1+1}$ and no other process has performed its $\mathsf{scan}$ of $S_{t_0+2}$
prior to $C_{j_1+1}$, all vertices in $\sigma$ have seen $a$. Thus the distance in $\mathbb{G}_{t_0+2}$
between $\sigma$ and $\mathbb{R}_{t_0+2}(a+1,a+2) \cup \mathbb{R}_{t_0+2}(a+2,a+3)$ is at least 1.
Since the distance from $\sigma$
to any terminated vertex in $\mathbb{G}_{t_0+2}$ is at least 2, the first edge on any path from
$\sigma$ to a terminated vertex or a vertex in $\mathbb{R}_{t_0+2}(a+1,a+2) \cup \mathbb{R}_{t_0+2}(a+2,a+3)$
is between active vertices. Therefore, by Lemma~\ref{lem:dist} and Proposition~\ref{prop:subdiv},
the distance in  $\mathbb{G}_{t_0+3}$ between $\chi(\sigma,\Delta)$ and $\chi(\mathbb{T}_{t_0+2}(b), \Delta) = \mathbb{T}_{t_0+3}(b)$ is at least 3 for every input $b$ such that $\mathbb{T}_{t_0+2}(b)$ is non-empty.
Similarly, the distance in  $\mathbb{G}_{t_0+3}$ between $\chi(\sigma,\Delta)$ and
$\chi(\mathbb{R}_{t_0+2}(a+1,a+2) \cup \mathbb{R}_{t_0+2}(a+2,a+3), \Delta) =
\mathbb{R}_{t_0+3}(a+1,a+2) \cup \mathbb{R}_{t_0+3}(a+2,a+3)$ is at least 2.

Consider the first $j_2 > j_1$ such that process $q_{j_2}$ is poised to $\mathsf{scan}$ the snapshot object $S_{t_0+3}$ in $C_{j_2}$. The states of $q_{j_2}$ in $C_{j_2+1}=C_{j_2}q_{j_2}$ is a vertex in $\chi(\sigma,\Delta)$.
According to its strategy for phase 1, the adversary terminates this vertex after query $(C_{j_2},q_{j_2})$. 
This contradicts the definition of $P$.
\end{proof}

\subsection{The adversarial strategy for later phases}
Since the prover does not win in phase 1, it must eventually choose a configuration $C \in \config{A}'(1)$ at the end of phase 1. The adversary will update $\Delta$ one final time. Afterwards, it can answer all future queries by the prover. The prover will eventually be forced to choose a terminal configuration at the end of some future phase and, consequently, will lose in the next phase.

Assume $C$ is a configuration reached by a non-empty schedule $\alpha(2)$ from an initial configuration $C_0 \in \config{A}(1)$. Let $p$ be the first  process in $\alpha(2)$ and let $a$ be its input in configuration $C_0$. 
Let $\mathbb{F}$ denote the union of all $n$-vertex cliques in $\mathbb{G}_1$ that represent configurations reachable by 
a 1-round schedule beginning with $p$ from configuration $C_0$ or an initial configuration that only
differs from $C_0$ by the states of processes that do not occur in $\alpha(2)$.
Since $p$ performs its $\mathsf{update}$ to $S_1$ before any process performs its $\mathsf{scan}$ of $S_1$ in all such schedules,
every vertex in $\mathbb{F}$ has seen $a$.
Thus the distance between $\mathbb{F}$ and 
$\mathbb{R}_{1}(a+1,a+2) \cup \mathbb{R}_{1}(a+2,a+3)$ in $\mathbb{G}_1$ is at least 1.

The adversary defines $\Delta(v) = \bot$ for each vertex $v$ in $\mathbb{G}_t$ where $\Delta(v)$ is undefined, subdivides $\St$ to construct $\Sto$, and increments $t$.
Since all the invariants hold at the end of phase 1, \lemmaref{lem:invariants-hold}
says that they still hold and,  for any two inputs $b \neq b'$ such that  $\Tt{b}$ and $\Tt{b'}$ are non-empty,
the distance between $\Tt{b}$ and $\Tt{b'}$ in $\St$ is at least 4.
In particular, a vertex $v$ in $\St$ is adjacent to a vertex $w \in \Tt{b}$ for at most one input $b$.
Let $\mathbb{F}' = \chi^{t-1}(\mathbb{F},\Delta) \subseteq \St$.
Applying \lemmaref{lem:invariants-hold} $t-1$ times,
it follows that the distance between $\mathbb{F}'$ and $\mathbb{R}_{t}(a+1,a+2) \cup \mathbb{R}_{t}(a+2,a+3)$ in $\St$ is at least 1.

Invariant~\ref{inv:definedt} says that no vertex in $\St$ has $\delta(v) = \bot$.
The adversary has not yet terminated any additional vertices in $\St$,
so, by Proposition~\ref{prop:subdiv},
$\Tt{b} = \mathbb{T}_{t-1}(b)$ for all input values $b$.
For every vertex $v \in \mathbb{F}'$ for which $\Delta(v)$ is undefined, the adversary defines $\Delta(v)$ as follows.
First, for each input value $b$
and each vertex $v \in \mathbb{F}'$  that is distance 1 from $\Tt{b}$ in $\St$ and 
such that $\Delta(v)$ is undefined,
the adversary sets $\Delta(v) = b$. 
By invariant~\ref{inv:contains}, the distance between $\Tt{b}$ and $\mathbb{R}_{t}(b+1,b+2) \cup \mathbb{R}_{t}(b+2,b+3)$ in $\St$ is at least 2. Thus setting $\Delta(v) = b$ does not violate validity.
Since each vertex in $\Xt{b}$ is at least distance 3 from any vertex in $\Tt{b}$,
this assignment does not contradict any output query that returned $\none$.
Moreover, the distance between any two vertices in $\mathbb{F}'$
that have output different values is still at least 2.
Thus, in each $n$-vertex simplex in $\St$, all the terminated vertices have output the same value. 

Finally,
for each vertex $v \in \mathbb{F}'$ where $\Delta(v)$ is still undefined, the adversary sets $\Delta(v) = a$. 
Validity is preserved, since no vertex in $\mathbb{F}'$ is in $\mathbb{R}_{t}(a+1,a+2) \cup \mathbb{R}_{t}(a+2,a+3)$.
Agreement is not violated, since at most two different values are output by the vertices
in each $n$-vertex simplex in $\St$.

In phases $\varphi \geq 2$, the prover can only query configurations reachable from some configuration in $\config{A}(2)$.
By definition, $\config{A}(2)$ is the set of all configurations that are reached by performing $\alpha(2)$ from initial configurations that only differ from $C_0$ by the states of processes that do not occur in $\alpha(2)$.
It follows that, for any process $q$ and 
any extension $\alpha'$ of $\alpha(2)$ from $C' \in \config{A}(2)$, $q$ appears at most $2t$ times in $\alpha(2)\alpha'$ before its state is represented by a vertex in $\mathbb{F}'$. By construction, every vertex in $\mathbb{F}'$ has terminated. 
Thus, eventually, the prover chooses a configuration at the end of some phase in which every process has terminated.
The prover loses in the next phase.

\section{Upper bounds for asynchronous systems}\label{sec:upper}

In this section, we provide upper bounds  
for graphical approximate agreement. 
We give
\begin{itemize}[noitemsep]
\item a 1-resilient algorithm on general graphs for $n \ge 2$ processes (\sectionref{sec:generalgraphs}), and
\item a wait-free algorithm on any nicely bridged graph for $n \ge 2$ processes (\sectionref{sec:bridgedgraphs}).
\end{itemize}

Let $G = (V,E)$ be a connected graph.
For any set $U \subseteq V$, the subgraph of $G$ induced by $U$ is the graph $G[U] = (U,F)$, where $F = \{ e \in E : e \subseteq U \}$. 
The distance between two vertices $u$ and $v$ in $G$ is denoted by $d(u,v)$.
The \emph{eccentricity} $\epsilon(v)$ of a node $v \in V$ is $\max \{ d(u,v) :  u \in V \}$. The \emph{diameter} of $G$ is $\diam(G) = \max \{ \epsilon(v) : v \in V \}$ and the \emph{radius} $\radius(G)$ of $G$ is $\min \{ \epsilon(v) : v \in V \}$. 
For any nonempty set $U \subseteq V$, let $D(U) = \max \{ d(u,v) : u,v \in U \}$. In particular, $\diam(G) = D(V)$. 

\subsection{A 1-resilient algorithm for general graphs}
\label{sec:generalgraphs}
Let $G = (V,E)$ be an arbitrary connected graph, for example, a $c$-cycle for some $c \ge 4$. We show that we can solve the approximate agreement problem on $G$ assuming at most one process crashes. Let $\diam(G)$ denote the diameter of $G$.
The intuitive idea of the algorithm is simple:
First use 2-set agreement to reduce the number of input values to at most 2
and then run approximate agreement on a path for $\lceil \log_2\diam(G) \rceil$ steps.

There is an easy 1-resilient algorithm for 2-set agreement.
However, the second step is not immediate, as there may be many paths of $G$ on which the approximate agreement algorithm could be run. However, since all processes know the graph $G$, we can avoid this difficulty by fixing in advance a shortest path between every pair of vertices. The rest of this section is dedicated to proving the following result.

\begin{theorem}\label{thm:1-resilient-alg}
  Let $G = (V,E)$ be a connected graph. Then for all $n \ge 2$, there exists a 1-resilient algorithm which solves approximate agreement on $G$.
\end{theorem}

\paragraph{Solving 2-set agreement.}
Fix a total order on $V$. For any nonempty subset $X \subseteq V$, let $\min(X)$ be the smallest element of $X$ under this order.
Let $x_i(0) \in V$ be the input of process $p_i$ and  let $T = \lceil \log_2\diam(G) \rceil$.
We will use a single-writer atomic snapshot object, $S_0$, whose components are initialised with the special value~$-$.   
Each process $p_i$:
\begin{itemize}[noitemsep]
\item performs $\mathsf{update}$ on the $i$th component of the snapshot object $S_0$, setting it to the value~$x_i(0)$,
\item repeatedly performs $\mathsf{scan}$ on the snapshot object $S_0$ until at least $n-1$ components have values other than $-$,
\item lets $X_i(0)$ be the \emph{set} of vertices returned by its last $\mathsf{scan}$, and
\item lets $x_i(1) = \min(X_i(0))$.
\end{itemize}

\paragraph{Approximate agreement on a path.}
For any two vertices $u,v \in V$, fix a shortest path between $u$ and $v$ in $G$ and let $g(u,v)$ be a fixed node
in the center of this path.
Then $d(u,g(u,v)), d(v,g(u,v)) \leq \lceil d(u,v)/2 \rceil$.
For any nonempty set $X \subseteq V$ of size at most two, define
$\psi(X)=u$ if $X = \{u\}$ and $\psi(X)=g(u,v)$ if $X=\{u,v\}$.
We will use a sequence $S_1, \ldots, S_{T}$ of single-writer atomic snapshot objects, whose components are initialised with the special value~$-$. 
For $t=1,\ldots, T$, each process $p_i$:
\begin{itemize}[noitemsep]
\item performs $\mathsf{update}$ on component $i$ of the snapshot object $S_t$, setting it to the vertex~$x_i(t)$,
\item repeatedly performs $\mathsf{scan}$ on the snapshot object $S_t$ until at least $n-1$ components have values other than $-$,
\item lets $X_i(t)$ be the \emph{set} of vertices returned by its last $\mathsf{scan}$, and
\item lets $x_i(t+1) = \psi(X_i(t))$.
\end{itemize}
The output of process $p_i$ is the value $x_i(T+1)$.

\paragraph{Correctness.}
Let $0 \le t \le T$.
If process $p_i$ crashes before 
computing $X_i(t)$,
we define $X_i(t)$ to be the empty set.
Observe that each process $p_i$ first performs $\mathsf{update}$ on $S_t$ with $x_i(t)$ before performing $\mathsf{scan}$ on $S_t$.
Thus, if $p_i$ computes $X_i(t)$, then $X_i(t)$ is nonempty. 

Each component of $S_t$ is $\mathsf{update}$d at most once.
Since $\mathsf{scan}$ is an atomic operation, the set of vertices returned in a $\mathsf{scan}$ is a subset of the set of vertices returned in any later $\mathsf{scan}$. 
Therefore, $X_i(t) \subseteq X_j(t)$ or $X_j(t) \subseteq X_i(t)$ for any $i$ and~$j$.
Each process continues performing $\mathsf{scan}$ until it crashes or $S_t$ contains at most one $-$.
Thus $\{ X_j(t) : 0 \leq j \leq n-1\}$ contains at most two nonempty sets.
Since $x_j(t+1)$ is a function of $X_j(t)$,
it follows that $\{ x_j(t+1) : X_j(t) \neq \emptyset\}$ contains at most two different vertices. 
These are the only values that are used to $\mathsf{update}$ components of $S_{t+1}$, so
$X_i(t+1) \subseteq \{ x_j(t+1) : X_j(t) \neq \emptyset \}$.
Hence, $|X_i(t+1)| \leq 2$ and, if $X_i(t+1) \neq \emptyset$, then $x_i(t+2)= \psi(X_i(t+1))$ is defined.

Let $X(t) =  \bigcup \{ X_i(t) : 0 \leq i < n\}$.
We use $X(T+1)$ to denote the set of output values.
Note that $X(t) \subseteq V$ for $0 \leq t \leq T+1$ and $X(0)$ is a subset of the input values.
If $t \geq 1$, then $X(t) \subseteq \{x_j(t) : X_j(t-1) \neq \emptyset \}$, so $|X(t)| \leq 2$.

\begin{restatable}{lemma}{shrink}
\label{lemma:path-distances-shrink}
Let $1 \le t \le T$. Then $D(X(t+1)) \le \lceil D(X(t)) / 2 \rceil$.
\end{restatable}
\begin{proof}
If $X_i(t) = X(t)$ for every nonempty set $X_i(t)$, then $x_i(t+1) = \psi(X(t))$.
Hence $X(t+1)$ will contain only one vertex and $D(X(t+1))=0$.
Otherwise, $X_i(t)$ is a nonempty, proper subset of
$X(t)$ for some $0 \leq i < n$.
Recall that $|X(t)| \leq 2$,
so $X_i(t) = \{u\}$ and $X(t) = \{u,v\}$ for some vertices $u \neq v$.
Since $X_j(t) \subseteq X_i(t)$ or 
$X_i(t) \subseteq X_j(t)$
for all $0 \leq j < n$, it follows that
every nonempty set $X_j(t)$ is either equal to
$\{ u\}$ or $\{ u,v\}$ and $x_j(t+1)$ is either
equal to $\psi(\{u\}) = u$ or $\psi( \{u,v\} ) = g(u,v)$.
 By definition of $g$, we have that $d(u,g(u,v)) \leq \lceil d(u,v)/2 \rceil$.
Since $X(t+1) \subseteq \{ u, g(u,v) \}$, it follows that $D(X(t+1)) \le \lceil D(X(t)) /2 \rceil$.
\end{proof}

\begin{proof}[Proof of \theoremref{thm:1-resilient-alg}.]
We verify that the agreement and validity properties of graphical approximate agreement are satisfied.
We proceed by induction to show that vertices in $X(t+1)$ lie on some shortest path between the values in $X(t)$ for all $1 \le t \le T$.
The case $t=1$ is true because $X(1) \subseteq X(0)$.
Suppose the claim holds for some $X(t)$ such that $1 \le t \le T$. 
By definition of $g$ and $\psi$,
all values in $X(t+1)$ lie on some shortest path between the values in~$X(t)$. Thus, validity is satisfied.
Since $X(1) \subseteq V$,
 $D(X(1)) \leq \diam(G)$. As $T = \lceil \log_2 \diam(G) \rceil$, \lemmaref{lemma:path-distances-shrink} implies that the distance $d(u,v)$ between any two output values $u,v \in X(T+1)$ is at most
\[
\max\{ d(u,v) : u,v \in X(T+1) \} = D(X(T+1)) \leq \lceil \diam(G)/2^T \rceil \leq 1. \qedhere
\]
\end{proof}

In \sectionref{sec:synchronous}, we extend the same algorithmic idea to the synchronous message-passing setting under crash faults.

\subsection{A wait-free asynchronous algorithm for nicely bridged graphs}\label{sec:bridgedgraphs}
\label{thm:wait-free-nicely-bridged}

\paragraph{Preliminaries.}
The \emph{center} of $G$ is the set $ \{ v \in V : \epsilon(v) = \radius(G) \}$ of nodes with minimum eccentricity in $G$. A graph $G$ is \emph{$k$-self-centered} if every vertex has eccentricity $k$. This means that every vertex is in the center of $G$ and $\diam(G) = \radius(G) = k$. 
A graph is \emph{chordal} if it does not contain any induced cycles of length greater than three.
The \emph{3-sun}, also known as the Haj\'os graph, is obtained from a triangle $\{u,v,w\}$ by subdividing each of its edges and connecting the resulting three vertices $\{x,y,z\}$ to be a clique. This graph is 2-self-centered and chordal.

A set $K \subseteq V$ of nodes is (shortest path) \emph{convex} if, for any $u,v \in K$, all nodes on all shortest paths between $u$ and $v$ are contained in $K$. For any $U \subseteq V$, the \emph{convex hull} $\langle U \rangle$ of $U$ is the smallest convex superset of $U$. If $A \subseteq B$, then $\langle A \rangle \subseteq \langle B \rangle$. 
A vertex $v$ is \emph{simplicial} in the graph $G$ if the neighbours of $v$ in $G$ form a clique. 

\paragraph{Bridged and nicely bridged graphs.}
A subgraph $H$ of $G$ is \emph{isometric} if the distances between any two vertices of $H$ are the same in $H$ and $G$. A graph is \emph{bridged} if it contains no isometric cycles of length greater than three~\cite{farber1989diameters}. All chordal graphs are bridged, but a bridged graph may contain induced cycles of length greater than five. We say that $G = (V,E)$ is \emph{nicely bridged} if any 2-self-centered subgraph $H = G[S]$, induced by a convex set $S \subseteq V$, is chordal. 
Chordal graphs, $3$-sun-free bridged graphs, and bridged graphs with no four cliques are examples of nicely bridged graphs.

We now list some useful properties of bridged graphs.
Farber gave the following result about the radius and diameter of bridged graphs \cite{farber1989diameters}.

\begin{lemma}\label{lemma:bridged-radius}
For any bridged graph $G$, we have $3 \cdot \radius(G) \le 2 \cdot \diam(G) + 2$.
  If $G$ is bridged and does not contain a 3-sun as an induced subgraph, then $2 \cdot \radius(G) \le \diam(G)+1$ holds.
\end{lemma}

We use the following fact due to Farber and Jamison~\cite[Theorem 6.5]{Farber1987Local}.
\begin{lemma}
\label{lemma:bridged-hull-diameter}
If $G=(V,E)$ is bridged, then $D(\langle U \rangle) = D(U)$ for any nonempty $U \subseteq V$.
\end{lemma}

Next, we prove the following simple lemma.

\begin{restatable}{lemma}{bridgedinduceddiameter}
\label{lemma:bridged-induced-diameter}
If $G = (V,E)$ is bridged and $H = G[\langle U \rangle]$ for $U \subseteq V$, then $\diam(H)~=~D(\langle U \rangle)$.
\end{restatable}
\begin{proof}
Let $u,v \in \langle U \rangle$.
Consider any shortest path between $u$ and $v$ in $G$. By definition of $\langle U \rangle$, all vertices on this path are in $\langle U \rangle$. Thus, this is also a path between $u$ and $v$ in $H$.
Since $H$ is an induced subgraph of $G$, any shortest path between $u$ and $v$ in $H$ is also a path between $u$ and $v$ in $G$.
Hence the distance between $u$ and $v$ in
$G$ is the same as the distance between $u$ and $v$ in $H$.
It follows that $D(\langle U \rangle) = \diam(H)$.
\end{proof}

Note that an induced subgraph of a bridged graph is not necessarily a bridged graph. For example, consider wheel graphs.
However, the subgraph of a bridged graph induced by a convex set is bridged.

\begin{restatable}{lemma}{inducedbridged}\label{lemma:induced-bridge}
Let $G = (V,E)$ be a bridged graph and $S \subseteq V$. Then the subgraph $G[\langle S \rangle]$ is bridged.
\end{restatable}
\begin{proof}
Let $H = G[\langle S \rangle]$.
Consider a cycle $C$ of length at least four in $H$. Since $C$ is a cycle in $G$ and $G$ is bridged, there exists vertices $u$ and $v$ in $C$ 
such that the distance between them in $G$ is less than
the distance between them in $C$.
Consider a shortest path between $u$ and $v$ in $G$.
The shortest path convex hull $\langle S \rangle$ contains
this path, since $u$ and $v$ are vertices of $\langle S \rangle$.
Thus $H = G[\langle S \rangle]$ also contains this path.
Hence $C$ is not isometric.
\end{proof}

\paragraph{The algorithm.}
For any nonempty set of vertices $X \subseteq V$,  we choose a vertex $\psi(X)$ from the subgraph $H$ induced by $\langle X \rangle$ as follows: If the center of $H$ contains a vertex that is non-simplicial in $H$, then let $\psi(X)$ be any such vertex. Otherwise, let $\psi(X)$ be any vertex in the center of $H$.
By definition, $\psi(X)$ has minimum eccentricity in the subgraph of $G$ induced by $\langle X \rangle$. 
Since $\psi(X)$ is a vertex in the convex hull of $X$, it is on some shortest path between two vertices in $X$. 

Let $x_i(0)$ be the input of process $p_i$ and let $T^* = \lceil \log_{3/2} \diam(G) \rceil + 1$. 
The processes communicate using a sequence $S_0, \ldots, S_T$ of single-writer snapshot objects, where $T = \max\{ |V|, T^* \}$.
In each iteration $t = 0, \ldots, T$, each process $p_i$:
\begin{itemize}[noitemsep]
\item performs $\mathsf{update}$ on the $i$th component of the snapshot object $S_t$, setting it to the vertex~$x_i(t)$,
\item performs $\mathsf{scan}$ on the snapshot object $S_t$,
\item defines $X_i(t)$ be the \emph{set} of vertices returned by its $\mathsf{scan}$, and
\item sets $x_i(t+1) = \psi(X_i(t))$.
\end{itemize}
Once $p_i$ has computed $x_i(T+1)$, the process outputs this vertex and terminates. 

\paragraph{Correctness.}
As before, if $p_i$ crashes before computing the set $X_i(t)$,
we define $X_i(t)$ to be the empty set. 
Let $X(t) = \bigcup \{ X_i(t) : 0 \leq i < n\}$. 
Note that $X(t) \subseteq V$ for $0 \leq t \leq T+1$ and $X(0)$ is a subset of the input values.
In particular, $X(t)$ is the set of values returned by the last $\mathsf{scan}$ performed on the snapshot object $S_t$.
Observe that if $y \in X(t)$, then $y = x_i(t)$ for some $0 \leq i < n$.
This is because each process $p_j$ that performs $\mathsf{update}$ on $S_{t}$ does so only with value $x_j(t)$. Thus, component $j$ of $S_t$ is either $x_j(t)$ or the special initial value $-$.

We use $X(T+1)$ to denote the set of output vertices.
We show that $X(T+1)$ satisfies agreement (all the values are contained in a clique) and validity (all the values are in the shortest path convex hull of the inputs) of approximate agreement on $G$. We start with validity. 

\begin{lemma}[Validity]\label{lemma:validity-of-bridged}
Let $0 \le t \le T+1$. Then $X(t) \subseteq \langle X(0) \rangle$.
\end{lemma}
\begin{proof}
We proceed by induction on $t$.
  For $t=0$, we have $X(0) \subseteq \langle X(0) \rangle$.
  Suppose the claim holds for some $0 \le t \le T$.
  Let $y \in X(t+1)$.
  Since $y = x_i(t+1) = \psi(X_i(t))$, for some $0 \le i < n$, the vertex $y$ is in
  $\langle X_i(t) \rangle$. 
  As $X_i(t) \subseteq X(t)$, it follows that $\langle X_i(t) \rangle \subseteq \langle X(t) \rangle$. 
  Thus, $y \in \langle X_i(t) \rangle \subseteq \langle X(t) \rangle$.  By the induction hypothesis, $X(t) \subseteq \langle X(0) \rangle$,
  so $y \in \langle X(t) \rangle \subseteq \langle X(0) \rangle$.
\end{proof}

We show that, for $0 \le t \le T^*$, if the set of values $X(t)$ does not form a clique, then the diameter of $X(t+1)$ is roughly half the diameter of $X(t)$.
Recall that $D(U)$ is the maximum distance in $G$ between the any two nodes in $U$.
Note that $D(U') \le D(U)$ for $U' \subseteq U$.
As in \sectionref{sec:generalgraphs}, $X_i(t) \subseteq X_j(t)$ or $X_i(t) \subseteq X_j(t)$ for $i,j \in\{0,\ldots,n-1\}$.

\begin{lemma}\label{lemma:diameter-shrinks-bridged}
Let $0 \le t \le T^*$. Then $D(X(t+1)) \le \frac{2}{3}( D( X(t)) + 1 )$. 
Moreover, if $\langle X(t) \rangle$ does not contain a 3-sun as an induced subgraph, then
$D(X(t+1)) \le \frac{1}{2}( D(X(t)) + 1)$.
\end{lemma}
\begin{proof}
Let $x_i(t+1), x_j(t+1) \in X(t+1)$. 
Recall that, by definition, $x_i(t+1) = \psi( X_i(t) ) \in \langle X_i(t) \rangle$ and $x_j(t+1) = \psi( X_j(t) ) \in \langle X_i(t) \rangle$.
Without loss of generality, assume that $X_j(t) \subseteq X_i(t)$.
Let $H$ be the subgraph of $G$ induced by $\langle X_i(t) \rangle$.
Since $G$ is bridged, the induced subgraph $H$ is also bridged, by \lemmaref{lemma:induced-bridge}. 
Since $X_j(t) \subseteq X_i(t) \subseteq \langle X_i(t) \rangle$, 
both $x_j(t+1)$ and $x_i(t+1)$ are vertices of~$H$.
Moreover, $X_i(t) \subseteq X(t)$ implies  
that $D(X_i(t)) \le D(X(t))$. By 
\lemmaref{lemma:bridged-hull-diameter} and \lemmaref{lemma:bridged-induced-diameter},
\[
  \diam(H) = D( \langle X_i(t) \rangle) = D( X_i(t) ) \le D(X(t)).
\]
By definition, $x_i(t+1) = \psi(X_i(t))$, which is a vertex in the center of $H$.
Hence, by \lemmaref{lemma:bridged-radius},
\[
d\left( x_i(t+1), x_j(t+1)\right) \le \radius\left( H \right) \le  \frac{2}{3}\left(\diam\left( H \right) + 1\right) \le \frac{2}{3}\left(D(X(t)) + 1\right). 
\]
For the second claim, if $\langle X(t) \rangle$ does not contain a 3-sun as an induced subgraph, \lemmaref{lemma:bridged-radius} yields 
\[
d\left( x_i(t+1), x_j(t+1)\right) \le \radius\left( H \right) \le  \frac{1}{2}\left(\diam\left( H \right) + 1\right) \le \frac{1}{2}\left(D(X(t)) + 1\right). \qedhere
\]
\end{proof}

We can apply \lemmaref{lemma:diameter-shrinks-bridged} repeatedly to ensure that we quickly end up in a subgraph with diameter at most two.

\begin{lemma}
\label{lemma:bridged-shrinkage}
The set 
$X(T^*)$
has diameter at most two.
\end{lemma}
\begin{proof}
  First, we show by induction that for all $0 \leq t \le T^*$, we have
  \[
  D(X(t)) \le \left(\frac{2}{3}\right)^t \left( \diam(G) - 2 \right) + 2.
  \]The base case $t=0$ is vacuous as the distance between any two vertices is at most the diameter $\diam(G)$. 
  For the inductive step, suppose the claim holds for some $0 \le t < T^*$.
  By \lemmaref{lemma:diameter-shrinks-bridged},
  \begin{align*}
    D( X(t+1) ) &\le \frac{2}{3}\left[ D(X(t)) + 1 \right] 
    \le \frac{2}{3}\left[ \left(\frac{2}{3}\right)^t( \diam(G) - 2) + 3 \right] 
    = \left(\frac{2}{3}\right)^{t+1}( \diam(G) - 2 ) + 2.
 \end{align*}
   Since $T^* = \lceil \log_{3/2}\diam(G) \rceil + 1$ and the diameter is an integer, we get that    
   \[
   D(X(T)^*) \le \left\lfloor \frac{2}{3 \diam(G)}\left(\diam(G)-2\right) + 2 \right\rfloor \le  \left\lfloor \frac{2}{3} + 2 \right\rfloor = 2. \qedhere
   \] 
\end{proof}

\begin{lemma}\label{lemma:radius-one}
  If the subgraph induced by $\langle X(t) \rangle$ has radius one, then $X(t+1)$ is a clique.
\end{lemma}
\begin{proof}
Let $x_i(t+1), x_j(t+1) \in X(t+1)$ and assume without loss of generality that $X_j(t) \subseteq X_i(t)$.
Now $\langle X_j(t) \rangle \subseteq \langle X_i(t) \rangle$.
Since $x_i(t+1) = \psi ( X_i(t) )$ is a vertex in the center of subgraph induced by $\langle X_i(t) \rangle$, it is  adjacent to the vertex $x_j(t+1)$. Hence, any two values in $X(t+1)$ are adjacent.
\end{proof}

Thus, after reaching a subgraph of radius 1, one more iteration suffices.
Moreover,
the algorithm solves the problem on any (possibly non-bridged) graph of radius one.
If the graph does not contain a 3-sun as an induced subgraph, then the algorithm converges in $T^*+1$ iterations. However, the above lemmas do not guarantee progress when the convex hull of $X(t)$ has diameter two and radius two. We handle this case next. 

\paragraph{Handling 2-self-centered graphs.}
In bridged graphs, the algorithm converges either to a clique or to a set whose convex hull induces a 2-self-centered subgraph. We show that if $G$ is nicely bridged, i.e., any 2-self-centered convex subgraph is chordal, our algorithm makes progress. However, our approach does not work for all bridged graphs, as there are non-chordal 2-self-centered bridged graphs which do not have any simplicial vertices. (For example, see
\appendixref{apx:nicely-bridged}.)

Recall that, if the center of $\langle X \rangle$ contains a non-simplicial vertex, then $\psi(X)$ is a non-simplicial vertex from the center of the subgraph induced by $\langle X \rangle$. This allows us to exclude simplicial vertices, which always exist in any chordal graph~\cite{dirac1961rigid}. By removing any simplicial vertex, the convex hull shrinks, as shown by the next lemma.

\begin{restatable}{lemma}{chordalremovingsimplicial}\label{lemma:chordal-removing-simplicial}
Let $U$ be a convex set. If $s \in U$ is simplicial in $U$, then 
 $\langle U \setminus \{s \} \rangle = U \setminus \{s \}$. 
\end{restatable}
\begin{proof}
Suppose $s \in \langle U \setminus \{s \}\rangle$. This means that there are two vertices $u,v \in U \setminus \{s \}$ such that the vertex~$s$ lies on some shortest path $u = w_0,, \ldots, w_k = v$ between $u$ and $v$.
Let $w_i = s$ for some $0 < i < k$.
Since $w_i$ is simplicial in $H$, the vertices $w_{i-1}$ and $w_{i+1}$ are adjacent in $H$. 
But now the path $w_0, \ldots, w_{i-1}, w_{i+1}, \ldots, w_k$ is a shorter path from $u$ to $v$, a contradiction.
\end{proof}

\begin{lemma} {\rm \cite{dirac1961rigid}} 
\label{lemma:simplicial-dirac}
Every chordal graph $G$ has a simplicial vertex.
If $G$ is not a clique, then it has two non-adjacent simplicial vertices.
\end{lemma}

\begin{lemma}\label{lemma:chordal-slow-shrink}
Let $0 \le t \le T$. If $D(X(t)) \ge 2$, then
$\langle X(t+1) \rangle \subsetneq \langle X(t) \rangle$.
\end{lemma}
\begin{proof}
Let $H$ be the subgraph induced by $\langle X(t) \rangle$.
Assume that $H$ has diameter and radius two; otherwise, the claim follows from  \lemmaref{lemma:bridged-shrinkage} and \lemmaref{lemma:radius-one}.
Since $G$ is nicely bridged, $H$ is chordal.
Let $S$ be the set of vertices that are simplicial in $H$.
Since $H$ is chordal and has diameter at least two, $S$ has two non-adjacent vertices, by \lemmaref{lemma:simplicial-dirac}. 
Let $S' = S \cap X(t+1)$. Observe that if $S' \subsetneq S$, then  \lemmaref{lemma:chordal-removing-simplicial} implies $\langle X(t+1) \rangle \subsetneq \langle X(t) \rangle$.
We show that the set $S' = S \cap X(t+1)$ is either empty or a clique, which implies that $S' \subsetneq S$.

For the sake of contradiction, let $x_i(t+1), x_j(t+1) \in S' \subseteq X(t+1)$ be two non-adjacent vertices.
We may assume that $X_j(t) \subseteq X_i(t)$, which implies that $\langle X_j(t) \rangle \subseteq \langle X_i(t) \rangle$.
Since $x_i(t+1)$ and $x_j(t+1)$ are non-adjacent, they are connected by a vertex $v \in \langle X_i(t) \rangle$.
Note that $v$ is not simplicial in $\langle X_i(t+1) \rangle$, but $x_i(t)$ is simplicial in $\langle X_i(t) \rangle$.
By definition, $x_i(t+1) = \psi( X_i(t) )$ is a vertex in the center of $\langle X_i(t) \rangle$, so it has eccentricity two. This means that $v$ also has eccentricity two and $v$ is also in the center. But, now, the center of $\langle X_i(t) \rangle$ contains a non-simplicial vertex $v$, which means that $x_i(t+1) = \psi( X_i(t) )$ is non-simplicial in $\langle X_i(t) \rangle$ by definition of $\psi$.
\end{proof}

\begin{proof}[Proof of \theoremref{thm:wait-free-nicely-bridged}.]
It remains to verify that agreement and validity conditions of graphical approximate agreement are satisfied. 
\lemmaref{lemma:validity-of-bridged} shows that validity is satisfied.
Repeated application of \lemmaref{lemma:chordal-slow-shrink} implies that $X(T+1)$ is a clique and, thus, agreement is satisfied. 
\end{proof}

\section{Upper bound for synchronous message-passing systems}\label{apx:synchronous}
\label{sec:synchronous}

Finally, we adapt the algorithm of \sectionref{sec:generalgraphs} to the synchronous message-passing setting under crash faults. This establishes the following upper bound, almost matching the lower bound given by \corollaryref{corollary:synchronous-lb}.

\begin{restatable}{theorem}{synchronous}
\label{thm:synchronous-upper}
Let $G$ be a connected graph.
For any $0 \le f < n$, there exists an $f$-resilient synchronous message-passing algorithm for $n$ processes that solves approximate agreement on $G$ in $\lfloor f/2 \rfloor + \lceil \log_2 \diam(G) \rceil + 1$ rounds.
\end{restatable}

To show this, we use the following result~\cite{chaudhuri2000tight}.
\begin{lemma}\label{lemma:synchronous-2-set}
For any $0 \le f < n$, there exists an $f$-resilient synchronous message-passing algorithm for $n$ processes that solves 2-set agreement in $\lfloor f/2 \rfloor + 1$ rounds.
\end{lemma}

\paragraph{Overview.}
The synchronous message-passing algorithm for graphical approximate agreement on $G$
follows the same idea as the asynchronous algorithm given in \sectionref{sec:generalgraphs}. All processes:
\begin{itemize}[noitemsep]
    \item use 2-set agreement to reduce the size of the set of inputs to at most 2, and then
    \item run approximate agreement on a path for $\lceil \log_2 \diam(G) \rceil$ steps.
\end{itemize}
By \lemmaref{lemma:synchronous-2-set}, the first part takes $\lfloor f/2 \rfloor + 1$ rounds.
By using similar arguments as in \sectionref{sec:generalgraphs}, we show that the second part takes $\lceil \log_2 \diam(G) \rceil$ rounds, and that the validity and agreement properties of graphical approximate agreement are satisfied.

\paragraph{Preliminaries.}
Let $x_i(0)$ be the input of process $p_i$ for the graphical approximate agreement problem.
As in \sectionref{sec:generalgraphs}, we let $g(u,v)$ be a fixed node in the center of a shortest path between $u$ and $v$, so $d(u,g(u,v)), d(v,g(u,v)) \leq \lceil d(u,v)/2 \rceil$ holds. For any nonempty set $X \subseteq V$ of size at most two, define
\[
\psi(X) = \begin{cases}
  u & \textrm{if } X = \{u\}, \\
  g(u,v) & \textrm{if } X = \{u,v\}
\end{cases}
\]

\paragraph{The algorithm.}
Let $T= \lceil \log_2 \diam(G) \rceil$. Each process $p_i$: 
\begin{enumerate}
    \item runs the 2-set agreement algorithm of \lemmaref{lemma:synchronous-2-set} with input $x_i(0)$ for $\lfloor f/2 \rfloor + 1$ rounds,
    \item lets $x_i(1)$ be its output in the 2-set agreement algorithm,
    \item for $t = 1, \ldots, T$ rounds,
    \begin{itemize}
        \item sends the value $x_i(t)$ to all processes in the system,
        \item receives a set $X_i(t)$ of values from other processes, 
        \item lets $x_i(t+1) = \psi(X_i(t))$, and
    \end{itemize}
    \item outputs the value $x_i(T+1)$.
\end{enumerate}
\paragraph{Correctness.}
The proof of correctness of the synchronous algorithm closely follows the proof of correctness of the asynchronous algorithm given in \sectionref{sec:generalgraphs}.
If process $p_i$ crashes before computing $X_i(t)$, we define $X_i(t)$ to be the empty set.
Let $X(t) = \bigcup \{ X_i(t) : 0 \leq i < n\}$
be the set of values received by any process
during the $t$th round of Step (3).
Note that each non-faulty process $p_i$ always sends the value $x_i(t)$ to itself, so $X_i(t)$ is nonempty if $p_i$ has not crashed by round $t$ of Step (3).  
We use $X(T+1)$ to denote the set of output values.

\begin{restatable}{lemma}{synchxsize}
\label{lemma:synchronous-xsize}
Let $1 \le t \le T$. Then $1 \le |X(t+1)| \le 2$.
\end{restatable}
\begin{proof}
We proceed by induction on $t$. 
For the base case $t=1$, observe that $|X(1)| \le 2$ holds by the agreement property of 2-set agreement.
For the inductive step, suppose that $1 \le |X(t)| \le 2$ holds for some $1 \le t \le T$.
Since $X_i(t) \subseteq X(t)$, it follows that $\{ X_i(t) : 0 \leq i \leq n-1\}$ contains at most two nonempty sets.
Moreover, since $x_i(t+1) = \psi(X_i(t))$, it follows that $X(t+1)$ will contain at most two different values.
\end{proof}

If the set of values $X(t)$ does not form a clique, then the diameter of $X(t+1)$ is roughly half the diameter of $X(t)$

\begin{lemma}\label{lemma:synchronous-path-distances-shrink}
Let $1 \le t \le T$. Then $D(X(t+1)) \le \lceil D(X(t)) / 2 \rceil$.
\end{lemma}
\begin{proof}
  By \lemmaref{lemma:synchronous-xsize}, we have $1 \le |X(t)| \le 2$ for $0 \le t \le T$. 
  There are two cases to consider.
  First, suppose $X(t) = \{ u \}$ for some node $u$. 
  If $X_i(t)$ is nonempty, then $X_i(t) = \{ u \}$ and $x_i(t+1) = \psi( \{u \} ) = u$.
  Hence, $X(t+1)$ will contain only node $u$ and $D(X(t+1)) = 0$.
  
 Next, suppose that $X(t) = \{ u,v \}$ for some nodes $u \neq v$. 
 The set $\{ X_i(t) : 0 \leq i \leq n-1\}$ contains at most two nonempty sets.
 Without loss of generality suppose that these are $\{u\}$ and $\{u,v\}$.
 If $X_i(t)$ is nonempty, then $x_i(t+1)$ is either $\psi( \{u \}) = u$ or $\psi( \{u,v\} ) = g(u,v)$.
 By the     definition of $g$, we have that $d(u,g(u,v)) \leq \lceil d(u,v)/2 \rceil$.
Since $X(t+1) \subseteq \{ u, g(u,v) \}$, it follows that $D(X(t+1) \le \lceil D(X(t)) /2 \rceil$.
\end{proof}

\begin{proof}[Proof of \theoremref{thm:synchronous-upper}]
By construction, the algorithm takes  
\[
\lfloor f/2 \rfloor + 1 + T  = \lfloor f/2 \rfloor + \lceil \log_2 \diam(G) \rceil + 1
\]
rounds.
Thus, we only need to verify that the outputs satisfy agreement and validity of graphical approximate agreement.

We proceed by induction to show that nodes in $X(t+1)$ lie on some shortest path between the values in $X(t)$ for all $0 \le t \le T$.
The case $t=1$ is true because by validity of 2-set agreement $X(1)$ consists only of initial input values. 
Suppose the claim holds for some $X(t)$ such that $0 \le t \le T$. 
By definition of $g$ and $\psi$,
all values in $X(t+1)$ lie on some shortest path between the values in $X(t)$. 
As shortest paths are also minimal paths, the set $X(T+1)$ of outputs satisfies validity. 
Now $D(X(0)) \leq \diam(G)$, since $X(0) \subseteq V$.
By \lemmaref{lemma:synchronous-path-distances-shrink}, the distance $d(u,v)$ between any two output values $u,v \in X(T+1)$ is at most
\[
\max\{ d(u,v) : u,v \in X(T+1) \} = D(X(T+1)) \leq \lceil \diam(G)/2^T \rceil \leq 1,
\]
since $T = \lceil \log_2 \diam(G) \rceil$.
\end{proof}

\section*{Acknowledgements}

We thank anonymous reviewers for their insightful comments and suggestions.
This project has received funding from the European Research Council (ERC) under the European Union's Horizon 2020 research and innovation programme (grant agreement No.\ 805223 ScaleML) and under the Marie Sk{\l}odowska-Curie grant agreement No.\ 840605
and from the Natural Science and Engineering Research Council of Canada grant RGPIN-2020-04178.

\bibliographystyle{plainnat}
\bibliography{references}

\appendix

\section{Examples of nicely bridged graphs \label{apx:nicely-bridged}}

In this section, we give some sufficient conditions for a bridged graph to be nicely bridged. If $C \cup \{x \}$ is an induced wheel of $G$, then we say that the wheel $C \cup \{x \}$ is \emph{uniquely centered} in $G$ if there is no $y \neq y$ such that $C \cup \{ y \}$ is also an induced wheel in $G$.

\begin{restatable}{theorem}{nicegraphs}
\label{thm:nicely-bridged-examples}
Let $G$ be a bridged graph. Then $G$ is nicely bridged if any of the following hold:
  \begin{enumerate}[label=(\alph*),noitemsep]
      \item $G$ is chordal.
      \item $G$ does not contain a 3-sun as an induced subgraph.
      \item Every wheel of $G$ is uniquely centered.
      \item $G$ has no cliques of size four.
  \end{enumerate}
\end{restatable}

The 3-sun is depicted in \figureref{fig:nice-and-not-so-nice}(a). It is chordal, and hence, nicely bridged.
\figureref{fig:nice-and-not-so-nice}(b) shows a bridged graph whose every wheel is uniquely centered. In contrast, \figureref{fig:nice-and-not-so-nice}(c) shows a bridged graph with a wheel that is not uniquely centered. This graph is also not nicely bridged. \figureref{fig:nice-and-not-so-nice}(d) gives an example of a bridged graph which has no simplicial vertices.

Recall that bridged graphs do not contain any induced cycles of length $4 \le k \le 5$, as every such cycle would be an isometric cycle of length at least four. In particular, any non-chordal bridged graph will have an induced cycle of length at least six. To establish \theoremref{thm:nicely-bridged-examples}, we start with the following lemma.

\begin{lemma}\label{lemma:2-sc-bridged-cycles}
  Suppose $G$ is a non-chordal bridged graph of diameter two.
  Let $C$ be a shortest induced $k$-cycle of length $k \ge 6$ in $G$.
  Then there exists a vertex $x$ such that $C \cup \{x \}$ induces a $k$-wheel in $G$.
\end{lemma}
\begin{proof}
Let $C = \{ c_1, \ldots, c_k \}$ be the shortest induced cycle of length $k \ge 6$.
Since $G$ has diameter two and $C$ is an induced cycle, we have $d(c_1,c_{k-2})=2$. Thus, there is some vertex $y$ connecting $c_1$ and $c_{k-2}$.
Since $G$ is bridged, then no subset of $\{ c_1, x, c_{k-2}, c_{k-1}, c_k \}$ can form an induced four or five cycle. Thus, so $x$ is adjacent to $c_k$ and $c_{k-1}$.
The $(k-1)$-cycle $\{ c_1, \ldots, c_{k-2}, x \}$ cannot be an induced cycle either, as the shortest induced cycle had length $k$.
This also implies that $y$ has to be adjacent to each $c_1, \ldots, c_{k-2}$ and $C \cup \{x \}$ induces a $k$-wheel. 
\end{proof}

\begin{lemma}\label{lemma:2-sc-cycle-adjacent}
  Suppose $G = (V,E)$ is a bridged graph that contains an induced cycle of length $k>3$. Let $C$ be the shortest such cycle and suppose $C \cup \{x \}$ induces a uniquely centered wheel. For any $v \in V$, let $A(v)$ be the neighbours of $v$ in $C$. If $A(v) \neq \emptyset$, then the following hold:
  \begin{enumerate}[noitemsep,label=(\alph*)]
  \item The set $A(v)$ induces a path of length at most three. 
  \item If $v$ is not adjacent to $x$, then $A(v)$ induces a path of length at most two.
  \end{enumerate}
\end{lemma}
\begin{proof}
Suppose $A(v)$ is nonempty and does not induce a path. Choose from $A(v)$ a pair of two such vertices $c_i$ and $c_j$ which have the shortest distance in $C$. Without loss of generality, we may assume these are the vertices $c_1$ and $c_i$ for some $2 < i \le k/2 + 1$. Now $\{v, c_1, \ldots, c_{i}\}$ induces a cycle of length $3 < i+1 \le k/2 + 2 < k$. This contradicts the fact that $k$ was the length of the shortest induced cycle of length at least four.

Next we show that the path induced by $A(v)$ has length at most three. Without loss of generality, assume that $A(v) = \{c_1, \ldots, c_h \}$. For the sake of contradiction, assume that $h > 3$.
Since $C$ is uniquely centered, $v$ cannot be adjacent to all vertices of $C$. Hence, this path has length $3 < h < k$.  Now $\{c_1, v, c_h, \ldots, c_{h+1}, \ldots c_k\}$ induces a cycle of length $3 < k - h + 2 < k -1$, which is a contradiction. 
For the last claim, observe that if $v$ is not adjacent to $x$ and $A(u) = \{c_1, \ldots, c_3\}$, then $\{x,c_1,v,c_3\}$ induces a four cycle.
\end{proof}

\begin{figure}
\begin{center}
  \includegraphics[page=3,width=0.85\textwidth]{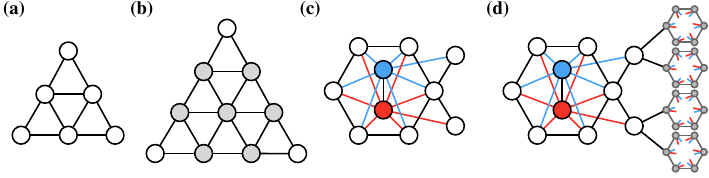}
  \caption{Examples of nicely bridged and not nicely bridged graphs. (a) The $3$-sun is chordal  and nicely bridged. (b) A nicely bridged graph. The grey vertices form a uniquely centered wheel. (c) A non-chordal 2-self-centered bridged graph. The cycle has a wheel that is not uniquely centered: both red and blue verticess are both axles of the wheel.
  (d) A 2-self-centered bridged graph with no simplicial vertices. Each grey vertex is connected to the red and blue vertex. \label{fig:nice-and-not-so-nice}}
\end{center}
\end{figure}

\begin{lemma}\label{lemma:2-sc-uniquely-centered}
  If $G$ is bridged and its every induced wheel is uniquely centered, then $G$ is nicely bridged.
\end{lemma}
\begin{proof}
Suppose $G$ is not nicely bridged, that is, there is a $H$ is 2-self-centered convex subgraph that is not chordal. By \lemmaref{lemma:induced-bridge} the graph $H$ is bridged. Since $H$ is bridged, but not chordal, $H$ contains some induced cycle of length at least six. Let $k > 5$ be the length of  the shortest induced cycle in $H$. Fix $C = \{ c_1, \ldots, c_k \}$ to be some induced cycle of length $k$. By \lemmaref{lemma:2-sc-bridged-cycles} there exists some vertex $x$ in $H$ such that $C \cup \{ x \}$ induces a $k$-wheel.

Because $H$ is 2-self-centered, $x$ has eccentricity two. Thus, there exists some vertex $y$ in $H$ such that $d(x,y)=2$. We show that the existence of such $y$ leads to the existence of a induced cycle of length 4, 5, or $k-1$, which contradicts the assumption that $C$ was the shortest cycle of length $k > 3$.
Clearly, $y \notin C$. By \lemmaref{lemma:2-sc-cycle-adjacent}, $y$ is adjacent to at most three consecutive vertices $A(y)$ of $C$. Without loss of generality, assume that $A(y) \subseteq \{c_1,c_2,c_3\}$. 

First, we show that $A(y)$ must be empty. Observe that $d(y,c_5) = 2$.
Hence, there is some vertex $v \notin C \cup \{x \}$ that is adjacent to both $y$ and $c_5$. 
However, $v$ can be adjacent to at most three consecutive vertices $A(v)$ of $C$. If $v$ is not adjacent to $c_3$, then either $\{v,y,c_3,c_4\}$ induces a four cycle or $\{v,y,c_3,c_4,c_5\}$ induces a five cycle. Thus, $v$ is adjacent to $c_3$. But then $\{c_1,y,z,c_5,\ldots, c_k\}$ induces a $(k-1)$-cycle. Thus, $A(y)$ must be empty.

Since $A(y)$ is empty, the vertex $y$ is not adjacent to any vertex $c_i \in C$. But since $H$ is 2-self-centered, $d(y,c_i)=2$ for all $c_i \in C$. Choose a neighbor $u$ of $y$ that is connected to $c_1$. By \lemmaref{lemma:2-sc-cycle-adjacent} the set $A(u)$ induces a path. Without loss of generality, assume that $A(u) = \{ c_1, \ldots, c_j\}$ for some $1 \le j \le 3$.  Since $d(y,c_5)=2$, there is some $v$ that is adjacent to $y$ and $c_5$. By \lemmaref{lemma:2-sc-cycle-adjacent} the set $A(v)$ induces a path of length at most three and so $A(v) \subseteq  \{ c_3, \ldots, c_k, c_1 \}$.
Note that $A(u) \cap A(v)$ can intersect either at $c_1$ or $c_3$, since $C$ has length $k \ge 6$. 

\begin{enumerate}
\item Consider the case $A(u) \cap A(v) \neq \{ c_1 \}$.
Now $A(v) \subseteq \{ c_3, \ldots, c_h \}$ for $5 \le h \le k$.
If $u$ and $v$ are adjacent, then 
$\{c_1,u,v,c_h, \ldots, c_k \}$ is an induced cycle of length at least four and less than $k - (h - 5) \le k$, which is a contradiction. Hence, $u$ and $v$ are not adjacent.
This means that a subset of $C' = \{ y, u, c_1, x, c_5, v \}$ induces a cycle of length at least four.
Since $G$ cannot have any induced cycles of length four or five, $C'$ must be an induced 6-cycle.
Since $C$ was the shortest induced cycle of length $k \ge 6$, it follows that $k=6$.
By \lemmaref{lemma:2-sc-bridged-cycles} there is some $z \neq x$ connected to all vertices of $C'$. Now $c_1, c_5 \in A(z)$. Since $A(z)$ is an induced path of length at most three, this implies that $c_6 \in A(z)$. Now either $\{ y, v, c_6, c_1, u \}$ is an induced 5-cycle or $\{ y, v, c_1, u \}$ is an induced 4-cycle, a contradiction.

\item Consider the case $A(u) \cap A(v) = \{ c_1 \}$. This means that $k=6$ and $A(v) = \{ c_5, c_6, c_1\}$ and $v$ is adjacent to $x$ by \lemmaref{lemma:2-sc-cycle-adjacent}.
Since $d(y,c_3)=2$, there is some $w$ adjacent to $y$ and $c_3$. Now $A(w)$ induces a path of length at most three. 
Suppose $w$ is not adjacent to $v$. Then either $\{ y, w, c_3, x, v \}$ or $\{ y, w, x, v \}$ is an induced cycle. Hence, $w$ is adjacent to $v$. This implies that $A(w) = \{ c_3, c_4, c_5 \}$, as otherwise we could find another induced cycle of length either four or five. Thus $w$ is adjacent to $x$ by \lemmaref{lemma:2-sc-cycle-adjacent}.

If $w$ is not adjacent to $u$, then either $\{ y, u, x, w \}$ or $\{ y, u, c_1, x, w \}$ is an induced cycle of length four or five, respectively. Thus, $w$ is adjacent to $u$. If $j \le 2$ we have that $\{u, c_j, \ldots, c_3, w \}$ is an induced cycle of length four or five. Hence $A(u) = \{ c_1, c_2, c_3 \}$. But then $\{ u, c_3, c_4, c_5, v \}$ is an induced 5-cycle, which is a contradiction. \qedhere
\end{enumerate}
\end{proof}

\begin{proof}[Proof of \theoremref{thm:nicely-bridged-examples}]
(a)~The claim follows from the fact that every induced subgraph of a chordal graph is also chordal. Hence, this also holds for any subgraph induced by a convex set.

(b)~Suppose $H$ is a diameter two subgraph of $G$ induced by a convex set $S$. By \lemmaref{lemma:induced-bridge} $H = G[S]$ is bridged. Since $G$ does not contain an induced 3-sun, neither does $H$. Thus, by \lemmaref{lemma:bridged-radius} we have $\radius(H) \le (\diam(H) + 1)/2 = 3/2$. Since the radius must be integral, $H$ has radius one, and cannot be 2-self-centered. Therefore, $G$ is nicely bridged.

(c)~This is the claim from \lemmaref{lemma:2-sc-uniquely-centered}.

(d) We show the claim by establishing that every wheel of $G$ is uniquely centered. Suppose there exists a $k$-wheel for $k>3$ that is not uniquely centered. If no such wheel exists, then $G$ is chordal and it follows from (a) that $G$ is also nicely bridged. Let $C = \{c_0, \ldots, c_{k-1}\}$ be the induced $k$-cycle forming the wheel and $x \neq y$ be two vertices such that $C \cup \{x \}$ and $C \cup \{y\}$ both induce a $k$-wheel. Note that $x$ and $y$ are not adjacent, as otherwise $\{c_0,c_1,x,y\}$ would be a clique of size four. But since  $x$ and $y$ are not adjacent, the set $\{x,c_0,y,c_2\}$ induces a four cycle, which contradicts the fact that $G$ was bridged.
\end{proof}

\end{document}